\documentclass[journal,11pt,onecolumn]{IEEEtran}
\usepackage[a4paper,includeheadfoot,margin=2.5cm]{geometry}
\usepackage[square, numbers, comma, sort&compress]{natbib}  
\usepackage{amsmath,amsfonts,amssymb,amscd,amsthm,xspace}
\usepackage{graphicx}
\usepackage{ifpdf}
\usepackage{epstopdf}
\usepackage{float}
\usepackage{bbm}
\usepackage{array}
\usepackage[scriptsize]{subfigure}
\graphicspath{{./Figures/}}  
\usepackage{setspace}
\usepackage{epstopdf}

\usepackage{setspace}
\usepackage{xcolor} 
\usepackage[hyperindex=true,
  pdftitle={Group Testing with Prior Statistics},
  pdfauthor={Tongxin Li}, colorlinks=true,
  linkcolor=blue!70!black,
  citecolor=green!50!black,
pagebackref=true,plainpages=false, pdfpagelabels]{hyperref} 
\usepackage[kerning=true]{microtype} 
\usepackage{cleveref}   

\newcommand{\tps}{\mathcal{N}}  
\newcommand{\ts}{\mathcal{S}}   
\newcommand{\is}{\mathcal{N}}   
\newcommand{\lf}{\Omega}        
\newcommand{\nt}{T}			
\newcommand{\np}{\mathit{n}}    
\newcommand{\ivt}{\mathbf{X}}   
\newcommand{\pvt}{\mathbf{p}}   
\newcommand{\hpvt}{\widehat{\textbf{p}}} 
\newcommand{\rvt}{\mathbf{B}}	
\newcommand{\ovt}{\mathbf{Y}}   
\newcommand{\et}{\mathbbm{P}_{e}}         
\newcommand{\uet}{\overline{\mathbbm{P}}_e}      
\newcommand{\expt}{\mathbbm{E}}	
\newcommand{\ent}{H}   
\newcommand{\mif}{I}   

\theoremstyle{definition}
\newtheorem{definition}{Definition}
\theoremstyle{theorem}
\newtheorem{theorem}{Theorem}
\newtheorem{lemma}{Lemma}
\newtheorem{corollary}{Corollary}

\begin{document}

\title{Group Testing with Prior Statistics }

\author{Tongxin Li, Chun Lam Chan, Wenhao Huang, Tarik Kaced, and Sidharth Jaggi}
\maketitle

\begin{abstract}

We consider a new group testing model wherein each item is a binary random variable defined by
an {\it a priori} probability of being {\it defective}. 
We assume that each probability is small and that items are independent, but not necessarily identically distributed.
The goal of  a {\it group testing} algorithm is to identify with high probability the subset of defectives via non-linear (disjunctive) binary measurements. 
Our main contributions are two classes of algorithms: (1) adaptive algorithms with tests based either on a maximum entropy principle, or on a Shannon-Fano/Huffman codes; (2) non-adaptive divide and conquer algorithms.
Under loose assumptions on prior statistics and with high probability, our algorithms only need a number of measurements that is
close to the information-theoretic entropy lower bound, up to an explicitly-calculated universal constant factor.
We provide simulations to support our results.

\end{abstract}


\maketitle

\section{Introduction}

The group testing model was first suggested by
Dorfman~\cite{dorfman1943detection} over sixty years ago, and has since spawned
a vast affiliated literature on theory and applications (see the
book~\cite{du1993combinatorial} for a survey). 
The classical version of the group testing problem is that of {\it combinatorial group testing} (CGT).
In that version, it is known that there are $d$ defective items in a population of size
$n$ (the common assumption is that $d=o(n)$).
Non-linear binary disjunctive group tests (OR operations) are allowed, for which a subset of items is tested, and the test outcome is $1$ if at least one item being tested is defective, and $0$ otherwise.
In that setting, if we allow an average probability of error of at most $\et>0$, the information
theoretic lower bound $(1- \et)\log_2 {{n}\choose{d}} =(1-\et)d
\log_2\left ( \frac{n}{d} \right ) + {\cal O} (d)$ on the total number of tests is
necessary for both adaptive and non-adaptive algorithms (see for
instance~\cite{chan2011non,malioutov2012boolean,du2006pooling}). 
Adaptive group testing schemes essentially meeting this bound are
known~\cite{goodrich2008improved};  non-adaptive algorithms that
meet this bound up to small multiplicative factors are also
known~\cite{chan2011non}. Some results from the studies on CGT
also readily carry over to probabilistic model in this paper (called {\it probabilistic group testing} (PGT))
wherein the $n$ items are defective i.i.d. with a small
probability~\cite{wolf1985born, wadayama2013analysis}.\footnote{In fact, in this paper, we consider a more general setting -- the $n$ items are independently but \textit{not} identically distributed.}

We focus on a model where the statistics on the likelihood of
any given item to be defective are available prior to the design of the testing
procedure.
The motivation comes from real-world examples. For instance, when
testing a large population for a given disease (Dorfman's original motivation
in~\cite{dorfman1943detection}), historical data on the prevalence of the
disease in specific sub-populations parametrized by age, gender, height ,weight, etc
are often available.
Specifically, in a population of size $n$, we denote the status of
whether the $i$th item is defective or not by whether a corresponding binary random
variable $X_{i}$ is $1$ or $0$. The length-$n$ binary vector ${\ivt} \in
\{0,1\}^n$ is the {\it population vector}, whose recovery is the objective of the
group testing algorithm. 

Our working hypothesis on the prior statistics is that items might 
have distinct {\it a priori} probabilities of being defective (non-identical)
and are independent.\footnote{It is true that even this model is still quite restrictive --
probabilistic models with finer structure, such
as correlation between ``neighbouring'' variables, or graph constraints, to
model the effect of geography or social structures are the
subject of ongoing investigation.}
The knowledge of prior statistics can reduce significantly the number of required test
in some scenarios. Consider the following -- given the probability vector
$(p_1,\ldots,p_n)$ one can compute the {\it expected number of defective items}
as $\mu$, defined as the sum $\sum_{i=1}^{n} p_i$ of the individual probabilities, and
in fact by standard statistical arguments~\cite{hagerup1990guided} this quantity
can even be ``concentrated'' (for large enough $n$ it can be shown that with
high probability the actual number of defective items is ``relatively close'' to
its expectation).
One might then na\"ively try to use existing
PGT algorithms, under the assumption that an upper bound for $d$,
the number of defectives, is given by $(1+\delta)\mu$, for some ``small''
$\delta$. 
An immediate issue of most PGT algorithms is that they assume that the prior statistics are, in one form or another, 
uniform -- each item is equally likely to be defective. It is therefore by no means
clear why those algorithms would have the same performance in our scenario (a
na\"ive translation of results would indicate high probability of recovery with
$c\mu\log_2(n)$ tests for some universal constant $c$). Indeed, proving that
such results do indeed translate, at least for one specific algorithm for the
``usual'' PGT model is an important module of our proof (see the proofs of Theorem~\ref{th:ModCoco} and Theorem~\ref{th:BlockCoco}).

Another issue is performance-related. In general, $c\mu\log_2(n)$
tests are not necessarily within a universal constant factor of the lower
bounds on the number of tests required for high probability recovery.
Indeed, a direct extension of known information-theoretic arguments 
(\cite{malioutov2012boolean, chan2011non,chan2012non}) show that a natural
lower bound corresponds to the entropy, $H(\ivt) = \sum_{i=1}^{n} h(p_i)$ where $ h(p_i)$ denotes the binary entropy function of the random
variable $X_i$. For the sake of completeness, Theorem~\ref{thm:low_bnd}
reproduces these arguments in our (non-uniform) probabilistic model.
It is not hard to construct distributions of $\ivt$ such that the ratio between $H(\ivt)$ and $\mu\log_2(n)$ is
arbitrarily large. 
Indeed, there are some extremal instances of distribution on the population
vector $\ivt$ where $\mu$ is constant but the entropy $\ent(\ivt)$ is
arbitrarily small.\footnote{Note that for our adaptive algorithms in Section~\ref{sec:2.2.1} and~\ref{sec:2.2.2}, the upper bound on the expected number of tests requires no restriction on the distribution of the population vector $\ivt$. Nevertheless, the concentration result of our adaptive algorithms (Theorem~\ref{th:LA}) and the upper bounds on the number of tests for the non-adaptive algorithms (Theorem~\ref{th:ModCoco} and~\ref{th:BlockCoco}) do require some loose conditions defined in Section~\ref{sec:pre}.}
Thus existing PGT algorithms might not be optimal and 
have their performance is not guaranteed 
under ``standard'' input assumptions.

\subsection{Related Work}

Some previous attempts to analyze models with prior statistic
information include an information-theoretic approach, shown in~\cite{torney1998optimizing}. Although the model is different and more restrictive, they provide optimal or sub-optimal algorithms for a certain choices of parameters. In particular,
they deal with a special type of prior information, where the universal set is
partitioned, and within a part, all subsets of a fixed given size are
uniformly distributed. For this slight generalization of traditional
group-testing, they prove  the existence of a non-adaptive algorithm whose average
performance is information-theoretically optimal, up to a small constant factor.

The PGT model, in which each item is defective i.i.d. $\mathrm{Bernoulli}(p)$, was considered in~\cite{dorfman1943detection,sobel1959group,wolf1985born}. In some of these works, an interesting subclass of adaptive group testing strategies called ``nested test plans'' were introduced. Nest test plans are constructed in a ``laminar'' manner, \textit{i.e.,} groups comprising tests are either proper subsets of prior groups, or disjoint. A recursive algorithm was given in~\cite{sobel1959group,wolf1985born} to design optimal nested test plans. However, the computational complexity of this scheme is exponential in $n$, and also does not result in an explicit bound on the number of tests required by the algorithm. A natural lower bound stated in~\cite{sobel1959group,wolf1985born} on the number of tests required equals $n h(p)$ -- indeed, this is related to the information-theoretic lower bound derived in Theorem~\ref{thm:low_bnd}.

A relationship between PGT and the Huffman codes was also mentioned in~\cite{wolf1985born} above. However, the Huffman-based design specified therein differs significantly from the one presented in this work -- the design in~\cite{wolf1985born} in general may result in group testing algorithms that require far more than the optimal number of test.

\subsection{Contributions}

We design explicit adaptive and non-adaptive group testing algorithms for the scenario with prior statistics on the probability of items being defective. In doing so, we discover intriguing and novel connections between source codes (such as Shannon-fano codes and Huffman codes) and adaptive sgroup testing algorithms. We prove that the expected number of tests required by both our adaptive and non-adaptive algorithms are information-theoretically optimal up to explicitly computed constant factors. Under mild assumptions on the probability distribution $(p_1,\ldots,p_n)$, we further prove that with high probability, the numbers of tests required for both the adaptive and non-adaptive algorithms are tightly concentrated around their expectations.

\section{Background}

\subsection{Preliminaries}
A summary of the notations used in this paper is given in Table~\ref{table:notation}.
\begin{table}[H]
	\caption{Nomenclature} 
	\label{table:notation} 
	\centering 
	\small
	\begin{tabular}{ l |p{12cm}} 
		\hline 
		Notation & Description\\ [0.5ex] 
		\hline \\[0.2ex]
		$\tps$ &\textit{Universal set} of all items being tested\\
		$L$ &Number of pre-partitioned subsets $\is_s$\\
		$\is_s$ &Disjoint \textit{pre-partitioned subsets} of the universal set $\tps$ indexed by $s=1,\ldots,L$\\
		
		$\lf$ &Laminar family of all tested subsets $\ts_{kl_{r}}$ in the adaptive algorithms\\
		
		$n$ &Total number of items, $n=|\tps|$\\ 
		
		$\nt$ & Number of tests used by the group-testing algorithm\\
		
		
		$M$ &A $\nt\times n$ Boolean matrix defining a group testing procedure\\ 
		
		$\ivt$ &Length-$n$ initial \textit{population vector} $(X_{1},...,X_{n})$ 
		\textit{population vector} where $X_{i}$ are independent binary variables\\
		
		$\rvt$ &Length-$\nt$ binary coded \textit{result vector} $(B_{1},...,B_{n})$ 
		where $b_{i}$ is the outcome of the corresponding group test\\
		
		$\ovt$ &Length-$n$ output \textit{recovery vector} $\left(Y_{1},...,Y_{n}\right)$ 
		decoded from the result vector $\rvt$\\
		
		$\pvt$ &Length-$n$ real-valued probability vector $(p_{1},...,p_{n})$ where
		$p_i$ is the \textit{a priori} probability of $x_i$ to be defective\\
		$\mu$  &Expected number of defective items $\mu$ defined by
	$\mu=\sum_{i=1}^{n}{p_i}$ \\
		$\widehat{\pvt}$ &Modified probability vector $\pvt=\left(\widehat{p}_{1},...,\widehat{p}_{n}\right)$ by letting
		$\widehat{p_{i}}:={\left(1-p_{i}\right)}/{\left(n-\mu\right)}$\\
		
		$\et$ &\textit{Probability of error} defined by $\et:=\Pr\left[\ivt\not=\ovt\right]$\\
		
		$m$ &Number of subsets in the same step of tests in the laminar family $\lf$
		and $k$ is the index used for the depth of the binary tree.  For example, the
		total number of subsets $S_{kl_{r}}$ in the $k$th stage is $m_{k}$ \\
		
		$g$ &\textit{Group testing sampling parameter} for the testing matrix $M$ in
		the non-adaptive algorithm\\
		[1ex] 
		\hline 
	\end{tabular} 
\end{table}
\subsubsection{Model and Notations}
\label{sec:model}

Let $\tps=\lbrace{X_1,X_2,\ldots,X_n}\rbrace$ denote the
\textit{universal set}, the set of $n$ items being tested where each $X_i\in\{0,1\}$ is
a binary random variable independent with the others. 
Let $\ivt=\left(X_1,X_2,\ldots,X_n\right)\in\lbrace{0,1}\rbrace^n$ 
be the \textit{population vector}, the initial vector for our group testing. 
Furthermore, we assume each \textit{testing item} $X_i$ can be \textit{defective} 
with \textit{a priori probability} $p_{i}\geq 0$ which means that $X_{i}$ takes value $1$ with probability $p_{i}$. Denote by $\pvt:=(p_1,\ldots,p_n)$ the corresponding \textit{probability vector} for the items in $\tps$.
A \textit{test} is based on a subset $\ts\subseteq\tps$. If one or more than one items in the subset being tested are defective (taking values $1$), then the test outcome is positive, otherwise it is negative. The \textit{testing procedure} is the collection of all tests.
Denote the corresponding coded vector by
$\rvt=\left(B_1,B_2,\ldots,B_\nt\right)\in\lbrace{0,1}\rbrace^\nt$, namely the
\textit{result vector} that contains the test results.  
The decoding process returns an output vector denoted
by $\ovt=\left(Y_1,Y_2,\ldots,Y_n\right)\in\lbrace{0,1}\rbrace^n$, which is
called the \textit{recovery vector}. 

In our probabilistic model, we choose to translate the sparsity requirement in CGT into the following natural \textit{sparse property}:
the \textit{expected number of defective items} $\mu$ satisfies
$\mu=\sum_{i=1}^{n}{p_i}\ll\np$. The \textit{average probability of error}\footnote{For simplicity we sometime use the term ``probability of error'' stand for $\et$ throughout the paper.} is defined by $\et:=\Pr\left(\ivt\neq \ovt\right)$ where the randomness is over the realizations of the population vector $\ivt$ and the testing algorithms. In our design, an upper bound $\uet$ on $\et$ is selected before testing the items. The upper bound $\uet$ serves as a fixed \textit{error threshold} such that the considered algorithm is guaranteed to have an average probability of error satisfying $\et\leq \uet$. 

{Group testing algorithms based on the previously introduced framework perform a sequence of measurements and guarantee that $\ovt$ matches $\ivt$ with high probability. The objective is to minimize the number of tests $\nt$, and meanwhile to guarantee the reconstruction of the population vector with a ``small'' probability of error.

\subsubsection{Pre-partition Model}
\label{sec:pre}
We define a partition  on the universal set $\tps$, wherein the subsets satisfy some requirements. These requirements will induce a slight modification of our algorithms (for the first stage, we will test the subsets from the partition). This technicality allows us to prove a relatively acceptable bound on the number of tests required for the algorithms. In the sequel, we define the aforementioned partition and the corresponding conditions.


A pre-partition is established by sorting the probabilities in $\pvt$ according to the following definitions:
\begin{definition}
	\label{def:2}
	For any non-empty subset $\is_s$ ($s=1,\ldots,n$) of the universal set $\tps$, we say $\is_s$ is a
	{\emph{well-balanced subset}} if the corresponding \textit{a priori} probabilities $p_{i}$ of
	items $X_{i}$ in $\is_s$ satisfy the following constraint:
	\begin{align*}
	{p_{i}}^{2}&\leq p_{j} \qquad \forall \ X_{i},X_{j}\in \is_s.
	\end{align*}
	
	Furthermore, if the following constraint is also satisfied for some positive constant $0<\gamma<1$, 
	\begin{align*}
	\frac{\uet}{2\np}<p_{i}<\gamma \qquad \forall \ X_{i}\in \is_s,
	\end{align*}
	we say $\is_s$ is a
	$\gamma$-{\emph{bounded subset}} (or for simplicity, \textit{bounded subset}). Furthermore, the subset $\is_s$ is said to be {\emph{bounded from below by
			${\uet}/{2\np}$}} if $p_i>{\uet}/{2\np}$; $\is_s$ is said to be {\emph{bounded from above by
			$\gamma$}} if $p_i<\gamma$. Otherwise $\is_s$ is said to be
	$\gamma$-{\emph{unbounded}}.
\end{definition}

The pre-partition is as follows.

\begin{itemize}
	\item First we sort all the \textit{a priori} probabilities in the probability vector $\pvt$.
	\item Then we divide the universal
	set $\tps=\bigcup_{s=1}^{L}\is_s$ into $L>0$ many disjoint subsets $\is_s$ ($s=1,\ldots,L$). Without loss of generality, we assume 
	\begin{align*}
	L:=\log_2\left(\log_{1/\gamma}\left(\frac{2\np}{\uet}\right)\right)+2\in\mathbbm{N}^+
	\end{align*}
	is a positive integer such that $\tps$ consists of $2$ unbounded subsets (representing the tails) and
	$L-2$ \textit{well-balanced and $\gamma$-bounded} subsets. The partition is constructed according to the following
	\begin{align*}
	p_{i}\in\begin{cases}
	\left[0,\frac{\uet}{2\np}\right], \quad &\text{for} \ \ X_{i}\in \is_1\\
	\left(\left(\frac{\uet}{2\np}\right)^{\left(\frac{1}{2}\right)^{s-2}},\left(\frac{\uet}{2\np}\right)^{\left(\frac{1}{2}\right)^{s-1}}\right], \quad &\text{for} \ X_{i}\in \is_s, s=2,\ldots,L-1\\
	\left(\gamma,1\right], \quad &\text{for} X_{i}\in \is_L
	\end{cases}.
	\end{align*}

\end{itemize}

Note that $L$ is determined by our chosen error
threshold $\uet$, the probability vector $\pvt$ and the population size $\np$. After creating the partition, we classify the $L$ subsets using the definition below.

\begin{definition}
\label{def:3}
For any non-empty well-balanced subset $\is_s$ ($s=1,\ldots, L$) of the universal set $\tps$, we say that it is
 {\emph{ample}} if the cardinality of the subset satisfies $|\is_s|\geq
\Gamma_\gamma$. Otherwise we say $\is_s$ is {\emph{not ample}}.

\end{definition}

Now for all pre-partitioned subsets $\is_s$, if $\is_s$ is ample, we regard it
as a feasible subset for the group testing and implement our group-testing
algorithms on each such subset separately. For
those subsets $\is_s$ which are not ample and the last subset $\is_L$ which is
not bounded above by $1>\gamma>0$, we combine them together and test all the
items in the combined set individually; for the first subset $\is_1$ which is
not bounded below by ${\et}/{2\np}$, we simply regard all items in $\is_1$ as
non-defective items without doing any test. Based on the
testing framework specified above (see Figure~\ref{fig:partition}), one can show corresponding upper
bounds on the number of tests for both adaptive and non-adaptive algorithms. The results are given in
Section~\ref{ss:theorems} with proofs provided in appendices.

 The flowchart in Figure~\ref{fig:partition}
demonstrates this procedure.

\begin{figure}[htbp]
	\centering
	\includegraphics[scale=0.4]{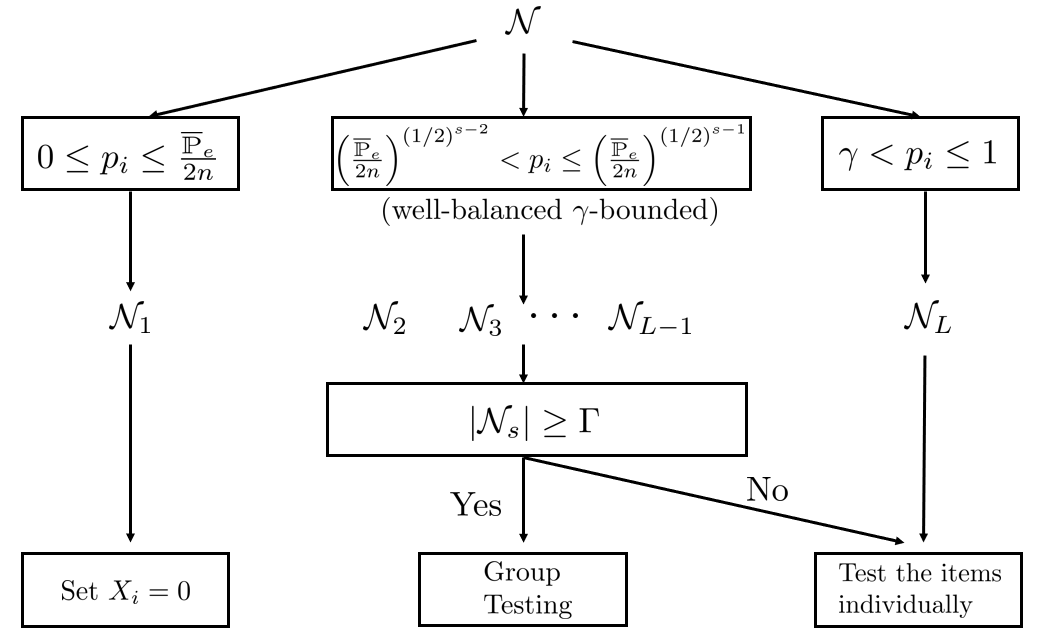}
	\caption{
		A diagram illustrating the partitioning procedure (for the non-adaptive algorithm). For instance, 
		two possible empty subsets $\is_1$ and $\is_L$ without bounds, while
		the remaining $L-2$ ``trimmed'' subsets are well-balanced and bounded. The elements in $\is_1$
		are assigned $0$ directly is because that the corresponding probabilities in $\is_1$ are small enough, thus applying the union bound, we can still get a proper upper bound on the error 
		probability. Details are given in Appendices \ref{pf2} and \ref{pf4}.
	}
	\label{fig:partition}
\end{figure}

\subsection{Fundamental Limits}

Before describing our algorithms, our first result states a universal
information-theoretic lower bound on the number of tests for the PGT model stated in Section~\ref{sec:model}.

\begin{theorem}[Lower Bound\footnote{Similar techniques were used in works in the Russian literature (see for instance \cite{Mal:78,Dya:04} to give information-theoretic lower bounds on the required number of tests when the probabilities of items being defective are homogeneous.}]
\label{thm:low_bnd}
Any Probabilistic Group Testing algorithm with noiseless measurements 
whose probability of error is at most $\et$ requires at least 
$\left(1-\et\right)\ent(\ivt)$ tests.
\end{theorem}

The proof can be found in Appendix~\ref{pf1}.

As an immediate corollary, if all probabilities $\{p_{i}\}_{i=1}^{n}$ are close to $1/2$, the
most efficient way to proceed is to test each element individually.\footnote{Considering the disjunctive nature of measurements, it is therefore natural to test the items in the tail set $\is_L$ individually.}

We believe that this theorem is a witness of a relationship between compression codes and group testing.
It is a counterpart of the well-known data compression
lower bound. Indeed, given a probability distribution, the expected length of any 
code is also bounded from below by the entropy $\ent(\ivt)$ of the distribution
$\ivt$. Further, sub-optimal/optimal codes such as Shannon-Fano/Huffman
codes~\cite{shannon2001mathematical,huffman1952method,cover2012elements} meet
this bound up to small additive factor.
Some of our algorithms also employ such codes, in a different way, and meet the
PGT lower bound in Theorem~\ref{thm:low_bnd} up to a multiplicative factor.

This theorem is also be used in Section~\ref{exp} as a benchmark for the
 simulations of the (adaptive and non-adaptive) algorithms.

\section{Main Results}

In the sequel, we formally describe the (adaptive and non-adaptive) algorithms for the considered PGT model.
\subsection{(Adaptive) Laminar Algorithms}
\label{sec:ALA}

In adaptive algorithms, the order of the tests matters since we can
design later tests according to the result of previous tests.
By design, our testing procedure will always satisfy the following property: if
a subset $\ts$ tested positive at stage $k$, then $\ts$ will be split into two
(children) subsets to be tested at stage $k+1$.
In this way, the whole testing procedure can be depicted as a tree where the number of stages
corresponds to depth. Child nodes correspond to subsets of items being
tested in their parent node. Leaves are individual tests, thus a path in this
tree identifies a single defective item.

Figure~\ref{fig:Co} partially exemplifies a typical structure of the testing tree described above. The depth of a tree represents the number of tests required, which is tightly related to the codeword length of a prefix-free code.
\begin{figure}[htbp]
	\centering
	\includegraphics[scale=0.43]{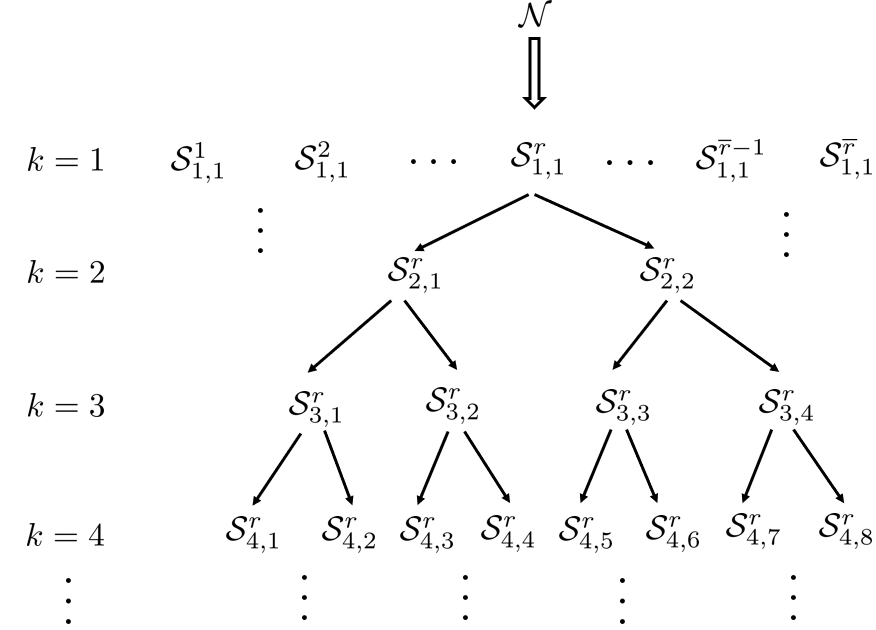}
	\caption[Co, n=512, fix t]{
		Graphical description of the laminar algorithm (without error
		estimation). The tests are done by forming trees for subsets of the universal set  $\tps$. First, we start with the initial set $\ivt$ containing $X_{1}$ up to
		$X_{n}$. Then we partition the universal set $\tps$ into several subsets from $\ts_{1,1}^{1}$ up to
		$\ts_{1,1}^{\overline{r}}$ satisfying that the quantity
		$\prod_{i\in\ts_{1,1}^{r}}{\left(1-p_{i}\right)}$ is close to $\frac{1}{2}$ for all $r=1,\ldots,\overline{r}$ and so on in this particular
		example. The construction of the testing trees can be done via two approaches -- maximum entropy-based approach and source codes-based approach, as introduced in Section~\ref{sec:2.2.1} and Section~\ref{sec:2.2.2} respectively.
	}
	\label{fig:Co}
\end{figure}

Next, we describe two different ways to construct the tree.
Both of them use a \textit{laminar family of subset} $\lf$
which contains subsets $\ts_{k,l}^{r}$ ($k,l$ are parameters indexing the child nodes and $r$ indexes the partitions) ~\cite{sobel1959group,wolf1985born}.
In this way, there is no cross-testing between
different trees. Each subset $\ts_{k,l}^{r}$ in the laminar family $\lf$ forms
a node in our set of testing trees as in Figure~\ref{fig:Co}.
Thus for simplicity, the constructed adaptive algorithms in this work are called \textit{laminar algorithms}.

We show that both two constructions in Section~\ref{sec:2.2.1} and Section~\ref{sec:2.2.2} achieve the same upper bound in Theorem~\ref{th:LA}.
A more detailed discussion is provided in Section~\ref{pf2}.

In the sequel, we formally describe the laminar algorithms.

\subsubsection{Maximum Entropy-based Laminar Algorithm}
\label{sec:2.2.1}
Given $k-1\leq\nt$, suppose we know the first $k-1$ outcomes
$\left(b_{1},b_{2},\dotsc,b_{k-1}\right)$ where $b_{i}'$ denotes the binary result of
test $i$. 
We define the next test by choosing a subset such that conditioned on the previous test outcomes, the probability $$\Pr\Big[B_{k}=0
|B_{k-1}=b_{k-1},\dotsc,B_{2}=b_{2},B_{1}=b_{1}\Big]$$ is close to
${1}/{2}$ (thus locally maximizing the information learned at each stage).

In general, getting a probability of exactly 1/2 is not possible due to the fact that the probability vector has arbitrary entries. Therefore, we choose the subsets being tested such that the probability they
contain a defective item is close to $1/2$, given the outcomes of the previous tests. Quantifying the impact of these ``quantization errors'', both in terms of the probability of error, and the number of tests required, is one of the major tasks in the proofs.
The algorithm is described and discussed in Section~\ref{sec:2.2.1} and the corresponding proof can be found in Appendix~\ref{pf2}.


Recall that the adaptive algorithms are ``tree-based''. The first stage of the tree\paragraph{First Stage}
is to divide the items into separate subsets. Indeed, the very first stage is based on an initial
partition. Thus the "tree" is not binary at the root but is binary afterwardsIn the first stage, we check whether the expected number of defectives $\mu=\sum_{i=1}^{n}{p_{i}}$ is smaller than the error threshold
(alternatively we can see it as a forest of binary trees). Indeed, each positive$\uet$. If so, we return $\ovt=\textbf{0}$;
test at stage $k$ induces two more (child) tests at stage $k+1$. otherwise, we  partition the universal set $\tps$ into subsets $\{\ts_{1,1}^{r}\}_{r=1}^{\overline{r}}$ in a \textit{greedy} manner as Figure~\ref{fig:Co} illustrates. \textit{I.e.,}
the partition is chosen such that $\Pr\left[\exists X_{i}\in \ts_{1,1}^{r}\
\text{s.t.}\ X_{i}=1\right]$ is the closest to ${1}/{2}$:
\begin{align}
&\min_{\ts_{1,1}^{r}}\left|\prod_{X_{i}\in \ts_{1,1}^{r}}{\left(1-p_{i}\right)}-\frac{1}{2}\right|\nonumber\\
&\mathrm{subject}\ \mathrm{to} \ \ts_{1,1}^{r}\subseteq\tps\backslash\bigcup_{j=1}^{r-1}\ts_{1,1}^{i}. \nonumber
\end{align}

\paragraph{Second Stage}

In the second stage, negative tests indicate that no item is defective. Thus, we only need to continue adaptively on those subsets with positive test outcomes. If a test is positive, for example, $B_r=1$, then
we divide the corresponding subset $\ts_{1,1}^{r}$ into two smaller subsets 
$\ts_{2,1}^{r}$, $\ts_{2,2}^{r}$ such that
$\Pr\left[\exists X_{i}\in \ts_{2,1}^{r}\ \text{s.t.}\
X_{i}=1|B_{r}=1\right]$ is the closest to ${1}/{2}$, i.e.,
\begin{align}
&\min_{\ts_{2,1}^{r}}\left|\frac{1-\prod_{X_{i}\in \ts_{2,1}^{r}}{\left(1-p_{i}\right)}}{1-\prod_{X_{i}\in \ts_{1,1}^{r}}{\left(1-p_{i}\right)}}-\frac{1}{2}\right|\nonumber\\
&\mathrm{subject}\ \mathrm{to} \ \ts_{2,1}^{r}\subseteq\ts_{1,1}^{r}.\nonumber
\end{align}

\paragraph{Later Stages}

Similarly, in the $k$th stage, we ignore the subsets that tested
negative in the previous stage (by marking the items inside non-defective), and split each of the remaining subsets\footnote{We only consider the subsets $\ts_{k,l}^{r}$ in which $l$ is an odd number, meaning that only the left nodes are considered and the right nodes are partitioned automatically.} into two parts
in a similar way: 
\begin{align}
&\min_{\ts_{k,l}^{r}}\left|\frac{1-\prod_{X_{i}\in \ts_{k,l}^{r}}{\left(1-p_{i}\right)}}{1-\prod_{X_{i}\in \ts_{k-1,l}^{r}}{\left(1-p_{i}\right)}}-\frac{1}{2}\right|\nonumber\\
&\mathrm{subject}\ \mathrm{to} \ \ts_{k,l}^{r}\subseteq\ts_{k-1,l}^{r}\nonumber
\end{align}
for all odd $l=1,3,\dotsc,2^{k-1}-1$.

Notice that, (a) we use ``contiguous'' partitions since the probability vector is
sorted; (b) all tests in a given stage involve disjoint subsets and can be thus
made in parallel; and (c) this procedure terminates and the leaves of the tree
correspond to tests on individual items.

\subsubsection{Shannon-Fano/Huffman Coding-based Laminar Algorithm}
\label{sec:2.2.2}

The second type of adaptive algorithms is based on the Shannon-Fano/Huffman source codes. Instead of greedily partitioning and constructing binary trees, an alternative choice is to use source codes.
Suppose the sum of probabilities in each
subset is less than one. Regarding the probabilities as weights, it is possible to construct corresponding Shannon–Fano or Huffman trees. We first partition the universal set $\tps$ into several
subsets.  Then regarding the probabilities $\{p_{i}\}_{i=1}^{n}$ as the corresponding ``weights'', testing trees can be constructed using Shannon-Fano/Huffman coding.

The construction is as follows:

\paragraph{First Stage}

The first stage is similar to the previous one except that we require the
product of $\left(1-p_{i}\right)$ in each subset to be strictly larger than
half. The partition satisfies:
\begin{align}
&\min\left|\prod_{X_{i}\in \ts_{1,1}^{r}}{\left(1-p_{i}\right)}-\frac{1}{2}\right|\nonumber\\
&\mathrm{subject}\ \mathrm{to} \  \nonumber\prod_{X_{i}\in \ts_{1,1}^{r}}{\left(1-p_{i}\right)}\in\left[\frac{1}{2},\frac{3}{4}\right]\\
&\qquad \qquad\quad\ts_{1,1}^{r}\subseteq\tps\backslash\bigcup_{j=1}^{r-1}\ts_{1,1}^{i}.\nonumber
\end{align}

\paragraph{Later Stages}
Next, within each subset $\ts_{1,1}^{r}$ ($r=1,\ldots,\overline{r}$)
we have $\prod_{x_{i}\in \ts_{1,1}^{r}}{\left(1-p_{i}\right)}\geq\frac{1}{2}$.
This implies that $\sum_{x_{i}\in \ts_{1,1}^{r}}{p_{i}}\leq 1$ (see Section~\ref{pf2} for the corresponding proof). 
For each subset $\ts_{1,1}^{r}$ ($r=1,\ldots,\overline{r}$), we set the weights $w_{i}$ 
as the corresponding $p_{i}$ and apply the Shannon-Fano coding, Huffman coding or
any source codes to construct the corresponding testing tree.

\subsubsection{Concentration of the Number of Tests}
In Theorem~\ref{th:LA}, we show that the expected number of tests can be bounded from above by $2\ent\left(\ivt\right)+2\mu$.
It remains to show that the actual number of steps in these algorithms is close to the
expected value, \textit{i.e.,} to concentrate the number of tests $\nt$ required. Since the items are independent, we partition the universal set $\tps$ into subsets and test the subsets individually using the aforementioned adaptive algorithms to guarantee the desired concentration results.  For more details, see Appendix~\ref{pf2}.



Next, we describe in details the non-adaptive block design.

\subsection{(Non-adaptive) Block Algorithm}
\label{B-Coco}

Non-adaptive algorithms require the testing procedure to be fixed in advance.
Therefore they may use more number of tests than adaptive algorithms, but the advantage is that the tests can be done in parallel, which is convenient for hardware design.

For the design of our non-adaptive algorithms, we represent the tests as a 
$T\times n$ Boolean matrix \textit{group-testing matrix} $\mathbf{M}\in\{0,1\}^{T\times n}$.
Each row of $\mathbf{M}$ corresponds to a measurement, and each
column corresponds to a single item to be tested. In this way, we have the
the population vector $\ivt\in\{0,1\}^n$ and the result vector $\rvt\in\{0,1\}^T$ satisfy
\begin{align*}
\rvt&=\mathbf{M} \ivt.
\end{align*}

\subsubsection{Coupon Collector Algorithm}

As introduced in~\cite{chan2012non}, the coupon collector algorithm (CCA) is a non-adaptive algorithm achieving the information-theoretic lower bound on the number of tests for the CGT model.
The cooresponding group-testing matrix $\mathbf{M}_{\mathrm{CCA}}$ is defined as follows.
A \textit{group testing sampling parameter $g$} is chosen by optimization 
which is fixed by the probability vector $\pvt$.
The $i$-th row of $M$ is then obtained by sampling
probability vector $\widehat{\pvt}=(\widehat{p}_1,\dotsc,\widehat{p}_n)$, where
$\widehat{p}_i={\left(1-p_i\right)}/{\left(n-\mu\right)}$, exactly $g$ times with replacement (for
convenience),
and setting $M_{i,j}=1$ if $j$ is sampled (at least once) during this process,
and zero otherwise~\cite{chan2012non}. The authors in~\cite{chan2012non} show that the testing procedure requires only 
$4(1+\delta)e\mu \ln n$ tests with high probability.

However, this bound is often worse than the corresponding information-theoretic lower bound in the PGT model, as the distribution $\pvt$ on items $x_{i}$ is sometime far from being
uniform.

\subsubsection{Block Design based on the Pre-partition Model}
\label{sec:block}
\begin{figure}[htbp]
	\centering
	\includegraphics[scale=0.43]{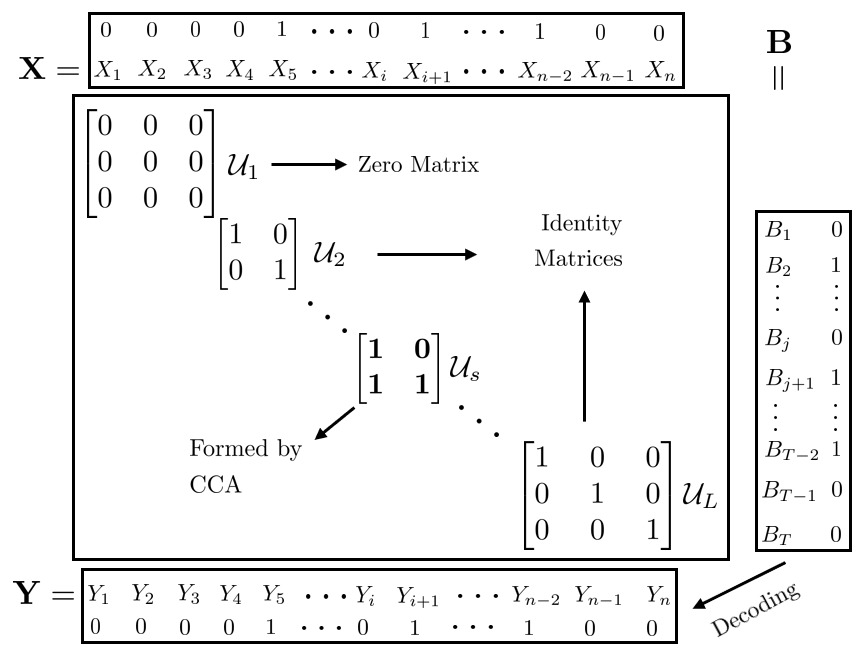}
	\caption[Co, n=512, fix t]{A typical group-testing matrix for the block algorithm introduced in Section~\ref{sec:block}. The encoding/decoding process contains three different cases. If a subset $\is_{s}$ is ample, CCA is implemented for the corresponding block (and the implementation is independent with other blocks); for $\is_1$, we simply set all items to $0$; moreover, for the subsets that are not ample or the subset is $\is_L$, we test the items inside individually.}
	\label{fig:block}
\end{figure}
In this work, we partially tackle this problem by employing the pre-partition model in Section~\ref{sec:pre} and designing a \textit{block algorithm} consisting of the CCA as a testing module for each block. Note that not all distributions $\pvt$ satisfy the definitions in Section~\ref{sec:pre}. Nonetheless, the definitions do cover a broad class of distributions that are nonuniform.

As in Section~\ref{sec:pre}, suppose there exist $L$ pre-partitioned subsets$\{\is_{s}\}_{s=1}^{L}$.
Then the group-testing matrix $\mathbf{M}$ for the whole testing procedure is
partitioned into $L-1$ sub-matrices $M_{s}$ corresponding to the subsets $\is_s$ that are bounded below by
${\et}/{2\np}$, and such that 
\begin{align*}
M&=\bigoplus_{s=2}^{L}{M_{s}}
\end{align*}
where $\bigoplus$ denotes the direct sum of matrices.

In our block algorithm, we assume the existence of suitable
pre-partition and use, as a sub-algorithm, the CCA for each subset $\is_{s}$, for which we control the testing complexity and the corresponding probability of error. We exemplify the matrix $\mathbf{M}$ of the block algorithm in
Figure~\ref{fig:block}.

Suppose the $L$ subsets are pre-partitioned according to the pre-partition model in Section~\ref{sec:pre}.
Each ample subset of the partition will be considered separately.
We use the following steps to specify the corresponding testing sub-matrix $M_{s}$ ($s=2,\ldots,L$):

First, according to the given \textit{a priori} probability vector $\pvt$, compute the corresponding $\widehat{\pvt}=(\widehat{p}_1,\widehat{p}_2,\ldots,\widehat{p}_n)$ where
\begin{align}
\label{eq:dis}
\widehat{p}_i=\left({1-p_i}\right)/\left({n-\mu}\right).
\end{align} Then compute the \textit{group testing sampling parameter $g^*$} by
\begin{align*}
g^{*}&:=\frac{-1}{\ln\left(\sum_{i=1}^{n}\widehat{p_{i}}\left(1-p_{i}\right)\right)},
\end{align*}
which is the optimal parameter for our purposes, as shown in~\ref{pf3}.

Then for each testing sub-matrix $M_{s}$, in each row we choose the items with
replacement $g^{*}$ times according to the probability distribution vector
$\widehat{\pvt}$ and form the testing matrix $M$ by
\begin{align*}
M&=\bigoplus_{s=2}^{L}{M_{s}}\nonumber
\end{align*}
as discussed in~\ref{B-Coco}.

\subsection{Upper Bounds on the Number of Tests $\nt$}
\label{ss:theorems}

For the adaptive group testing, the laminar algorithms introduced in Section~\ref{sec:ALA} satisfies the following theorem:
\begin{theorem}
	\label{th:LA} 

	The laminar algorithms (either maximum entropy-based or source codes-based) need at most $2\ent\left(\ivt\right)+2\mu$ tests in expectation, \textit{i.e.,}
	\begin{align*}
	\mathbbm{E}\left[T\right]\leq 2\ent\left(\ivt\right)+6\mu.
	\end{align*}
\end{theorem}	
Note that the items in $\tps$ are distributed independently. Partitioning the population set $\tps$ and testing the subsets individually, proper concentration inequalities imply the following corollary:
\begin{corollary}	
	\label{th:LA2} 
	Furthermore, with probability of error at least
	\begin{align*}
 1-\exp\left(-2\delta^2n^{1/4}\right).
	\end{align*}
the number of tests $\nt$ satisfies
	\begin{align*}
	\nt\leq 2\left(1+\delta\right) \left(\ent\left(\ivt\right)+3\mu\right).
	\end{align*}
\end{corollary}

For the non-adaptive group testing, the CCA introduced in~\cite{chan2012non} satisfies the following theorem:

\begin{theorem}[CCA~\cite{chan2012non}]
	\label{th:ModCoco}
	If the universal set $\tps$ is bounded from above by ${1}/{2}$, then for any $\delta>0$ the CCA in~\cite{chan2012non} requires no more than $$T\leq 4e\left(1+\delta\right)\mu \ln n$$ number of tests with probability of error at most
	$\et\leq 2n^{-\delta}$.
\end{theorem}

Furthermore, the block algorithm introduced in Section~\ref{sec:block} satisfies the following:

\begin{theorem}
	\label{th:BlockCoco}
	For any $0<\uet\leq 1$ and $\delta>0$, if the entropy of $\ivt$ satisfies
	\begin{align*}
	\ent\left({\ivt}\right)\geq  \Gamma_\gamma^2
	\end{align*}
	where 
	\begin{align*}
	\Gamma_\gamma:=\log_{2}\left(\log_{1/\gamma}\left(\frac{2\np}{\uet}\right)\right),
	\end{align*} 
	then with probability of error at most 
	\begin{align*}
	\et\leq{\Gamma_\gamma^{-\delta+1}}+\frac{1}{2}\uet,
	\end{align*}
	the block algorithm
	requires no more than $$\nt\leq \frac{e\ln n}{\log_2(1/\gamma)}\left(1+\delta\right)\ent\left(\ivt\right)+\Gamma_{\gamma}^2+ 2\mu$$ tests.
\end{theorem}

The proofs of Theorem~\ref{th:LA}, Corollary~\ref{th:LA2}, Theorem~\ref{th:ModCoco} and Theorem~\ref{th:BlockCoco} can be found in Appendix~\ref{pf2}, \ref{pf2.5}, \ref{pf3} and~\ref{pf4} respectively.

\section{Experimental Results}
\label{exp}
We provide experimental results for both the \textit{laminar algorithms} (LA) and the
\textit{block algorithm} (BA). 
We consider three different extremal types of probability vectors $\pvt$ -- \textit{uniform}, \textit{linear}, and \textit{exponential}.

For LA, both ME and Huffman constructions are tested. 
We used $200$ different points of entropy $\ent(\ivt)$.
As a result of Monte Carlo simulation, it is observed that the expected number of tests $\expt\left[\nt\right]$ computed from $200$ independent trials at each entropy
point grows linearly in $\ent(\ivt)$ where the coefficient is a positive constant as shown in~Figure~\ref{fig:bt-exp},~\ref{fig:bt-lin},~\ref{fig:bt-uni},~\ref{fig:hf-exp},~\ref{fig:hf-lin} and \ref{fig:hf-uni}. Moreover, the tests are for standard LA without using pre-partition model to ensure the concentration results.

For BA, 
based on the aforementioned three types of distributions of $\pvt$, we test three
different values of the expected number of defectives $\mu$ and compute the corresponding probabilities of error using $200$ independent trials. 
We compare the simulated probability of error with the theoretic probability of error 
$\et$  in~Figure~\ref{fig:b-exp},~\ref{fig:b-lin} and \ref{fig:b-uni}.

\clearpage
\begin{figure}
		\centering
		\subfigure[]{
					\centering
		\includegraphics[scale=0.245]{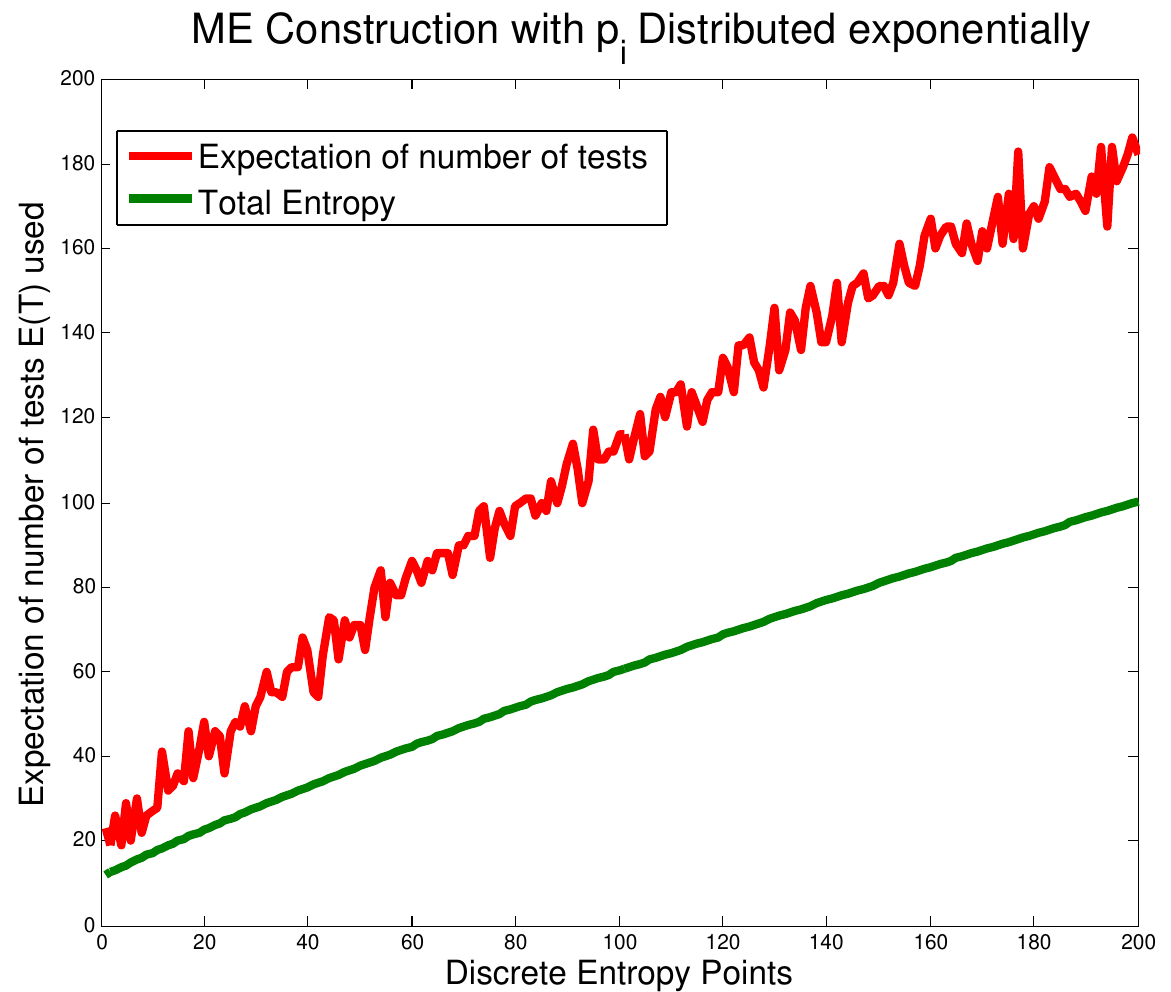}
	\label{fig:bt-exp}
}
\subfigure[]{
			\centering
		\includegraphics[scale=0.245]{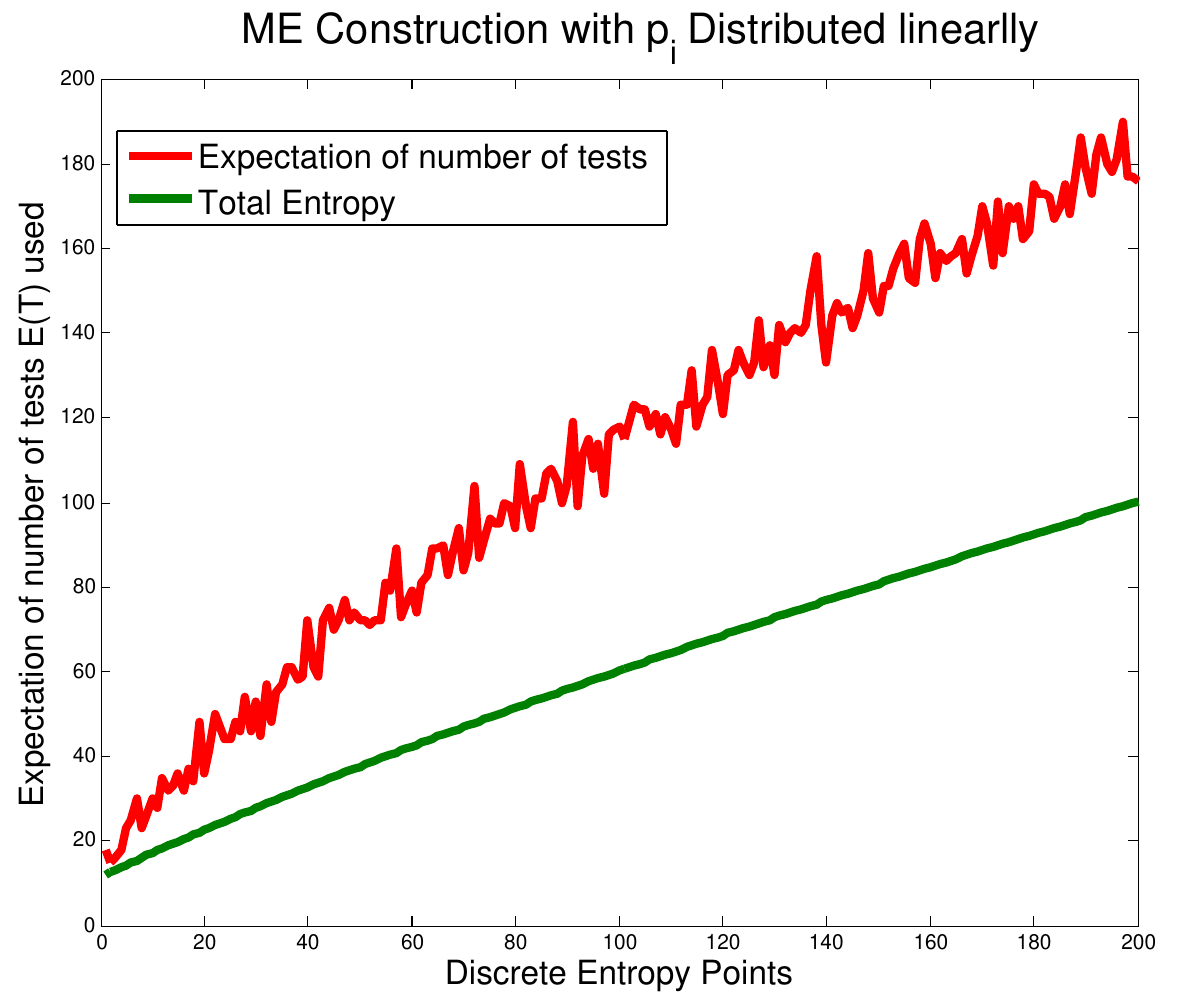}
	\label{fig:bt-lin}}
\subfigure[]{
			\centering
		\includegraphics[scale=0.245]{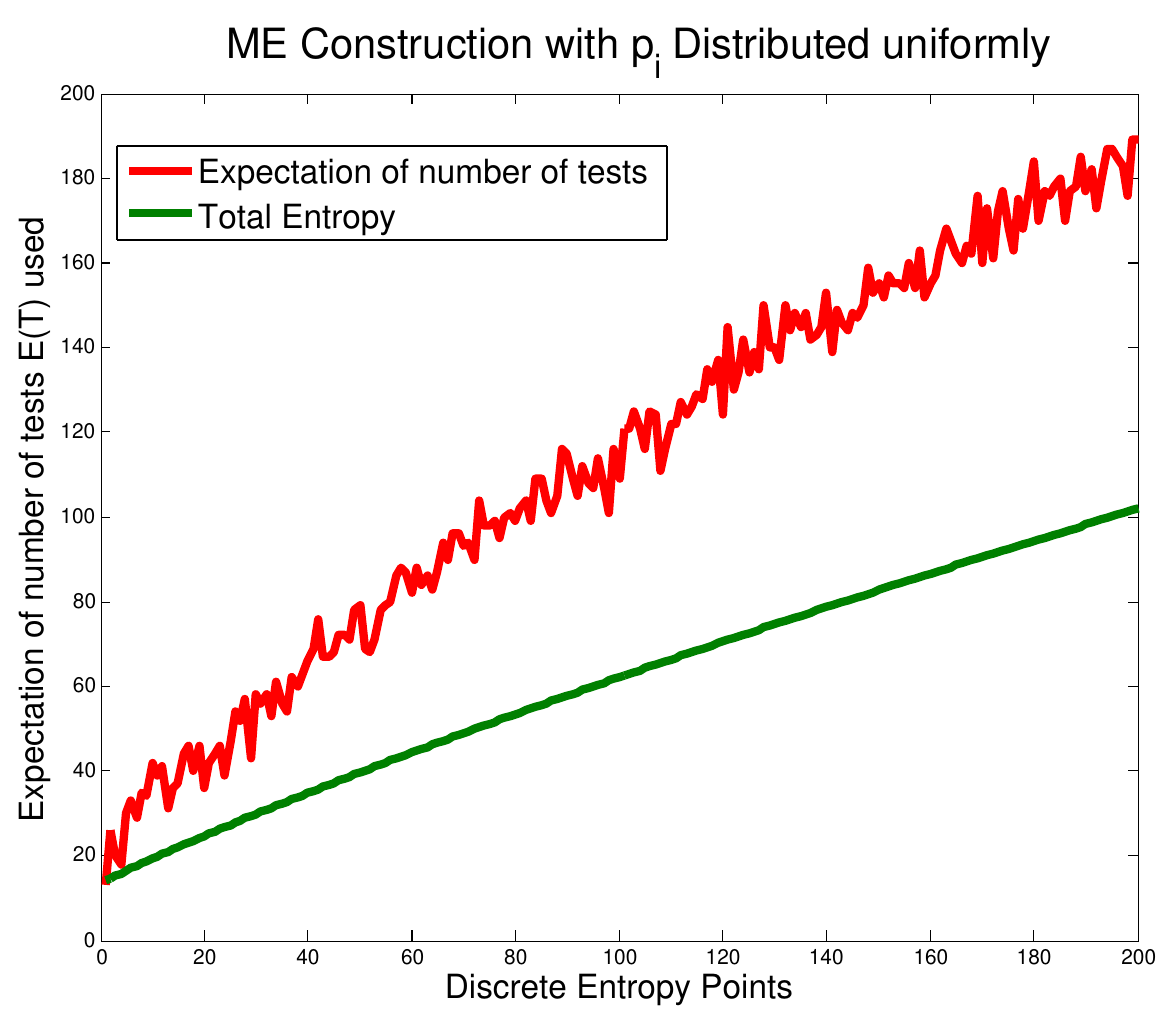}
	\label{fig:bt-uni}}
\caption{(a). Exponentially distributed $\pvt$ with $n=1000$ at $200$ different entropy points with $\expt\left[\nt\right]$ on each point calculated by $200$ independent trials. (b). Linearly distributed $\pvt$ with $n=1000$ at $200$ different entropy points with $\expt\left[\nt\right]$ on each point calculated by $200$ independent trials. (c). Uniformly distributed $\pvt$ with $n=1000$ at $200$ different entropy points with $\expt\left[\nt\right]$ on each point calculated by $200$ independent trials.}
\end{figure}
\begin{figure}
	\centering
	\subfigure[]{
		\includegraphics[scale=0.25]{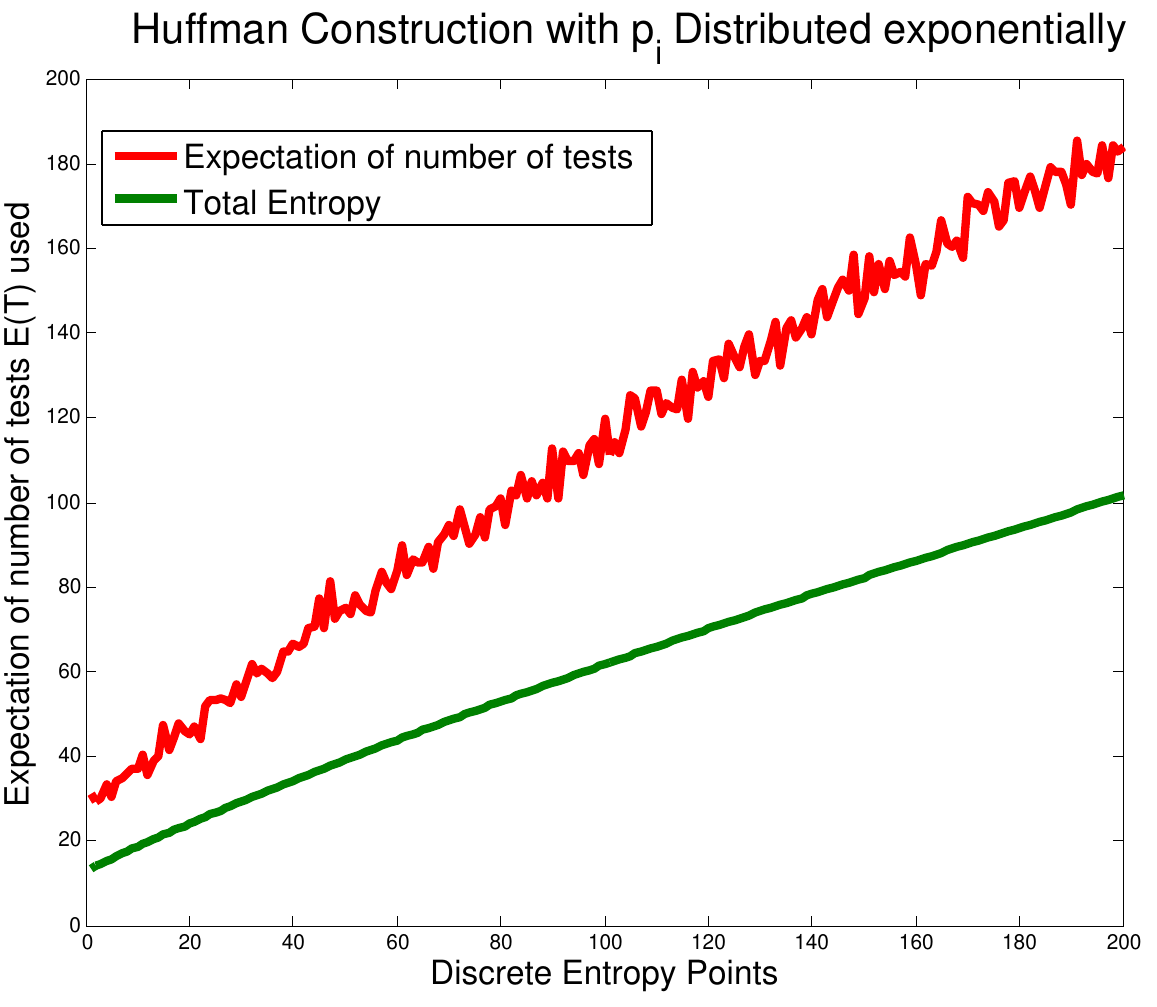}
	\label{fig:hf-exp}}
\subfigure[]{
		\includegraphics[scale=0.25]{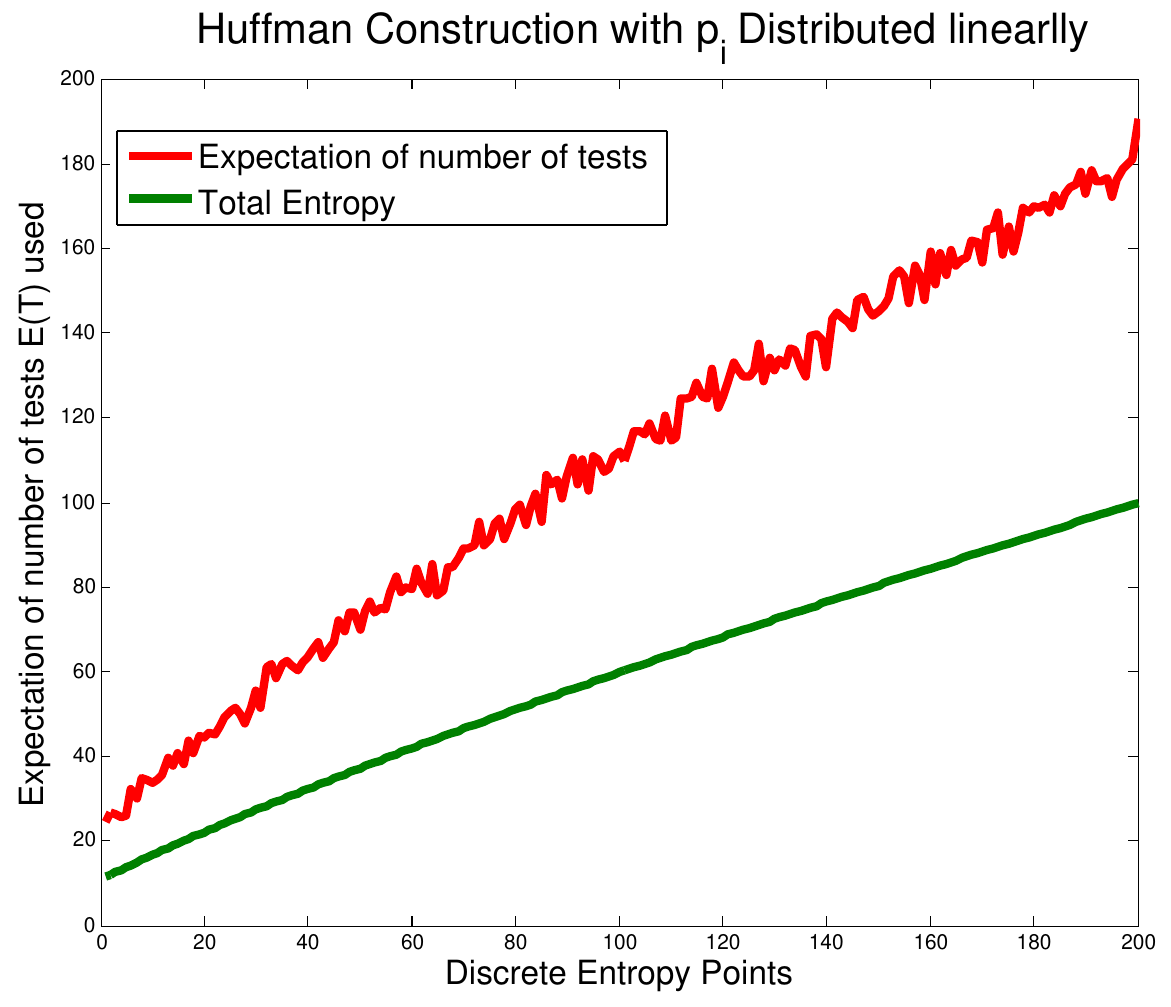}
	\label{fig:hf-lin}}
\subfigure[]{
		\includegraphics[scale=0.25]{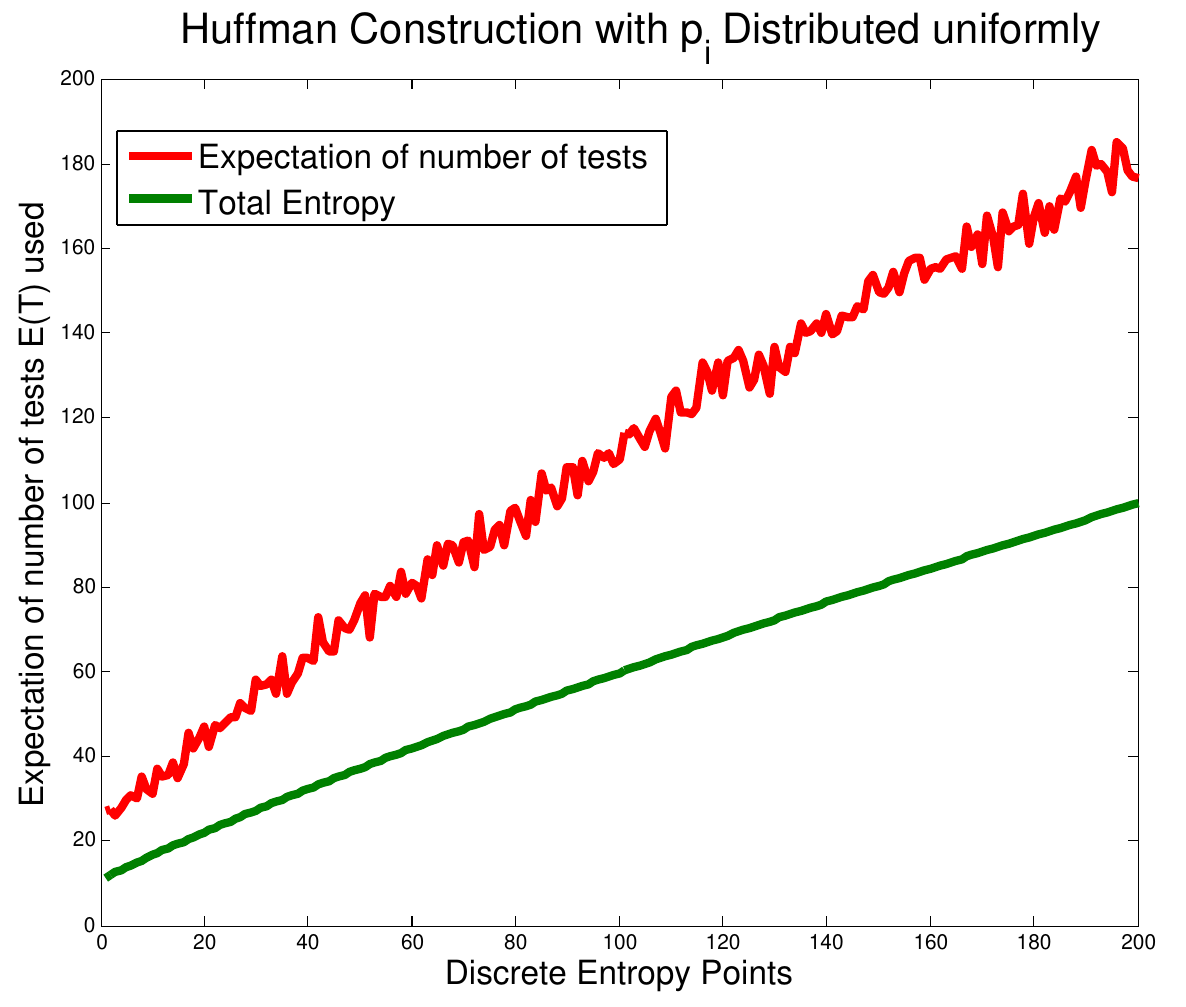}
	\label{fig:hf-uni}}
	\caption{(a). Exponentially distributed $\pvt$ with n=1000 at 200
		different entropy points with $E(T)$ on each point calculated by 200 independent trials. (b). Linearly distributed $\pvt$ with $n=1000$ at $200$ different entropy points with $\expt\left[\nt\right]$ on each point calculated by 200 independent trials. (c). Uniformly distributed $\pvt$ with $n=1000$ at $200$ different entropy points with $\expt\left[\nt\right]$ on each point calculated by $200$ independent trials.}
\end{figure}
\begin{figure}
	\centering
\subfigure[]{
		\includegraphics[scale=0.235]{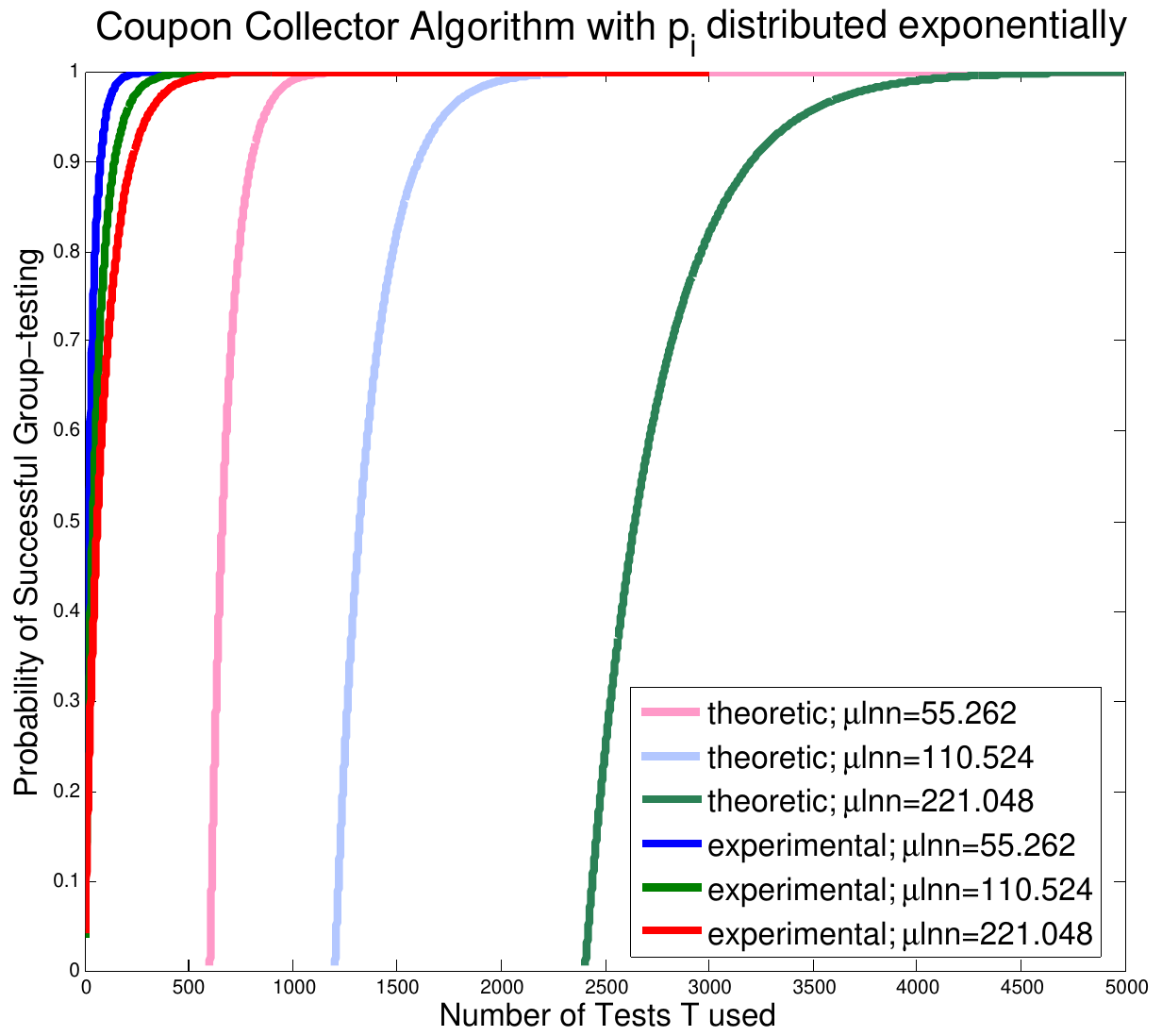}
	\label{fig:b-exp}}
\subfigure[]{
	\centering
		\includegraphics[scale=0.228]{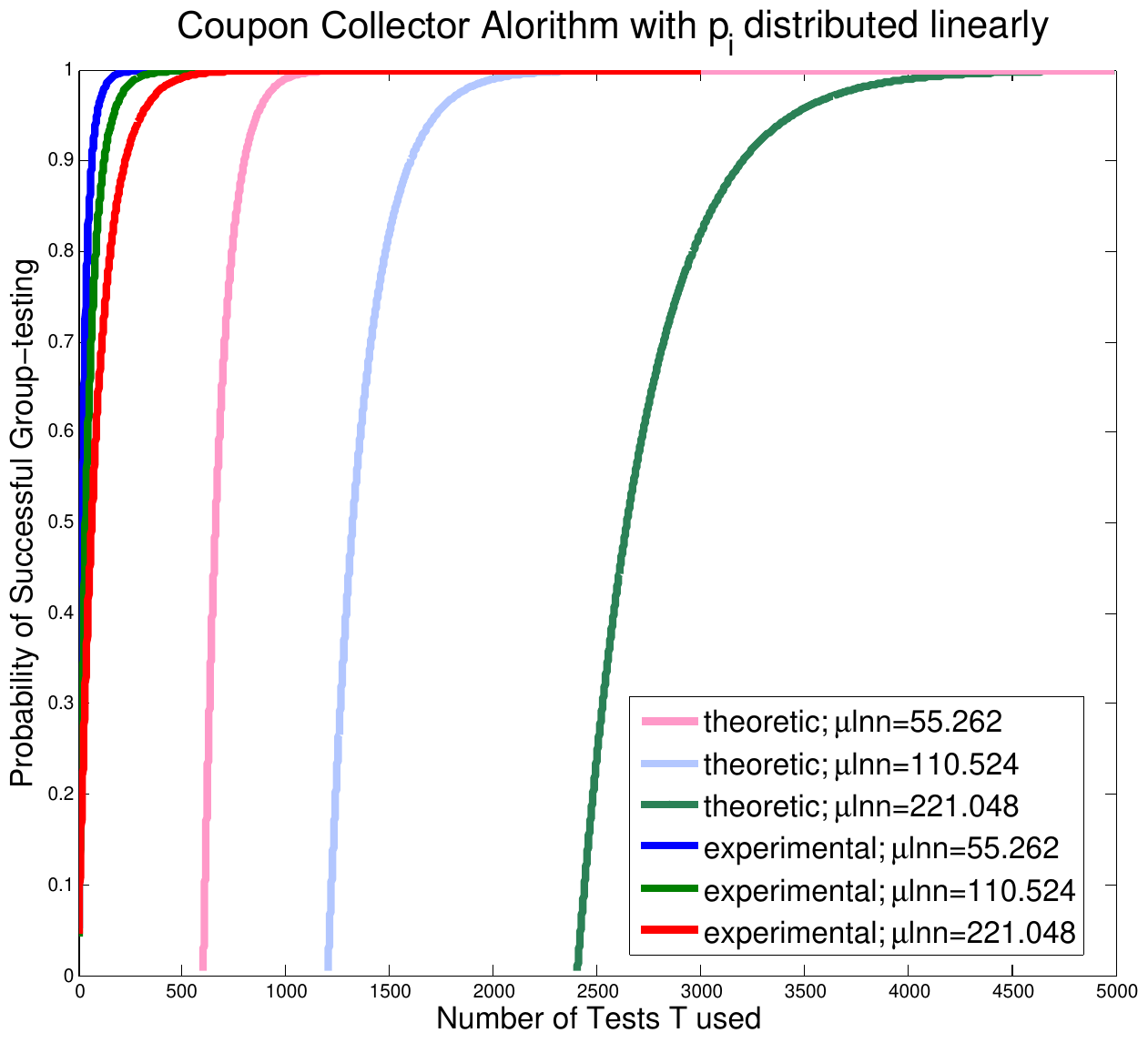}
	\label{fig:b-lin}}
\subfigure[]{
	\centering
		\includegraphics[scale=0.235]{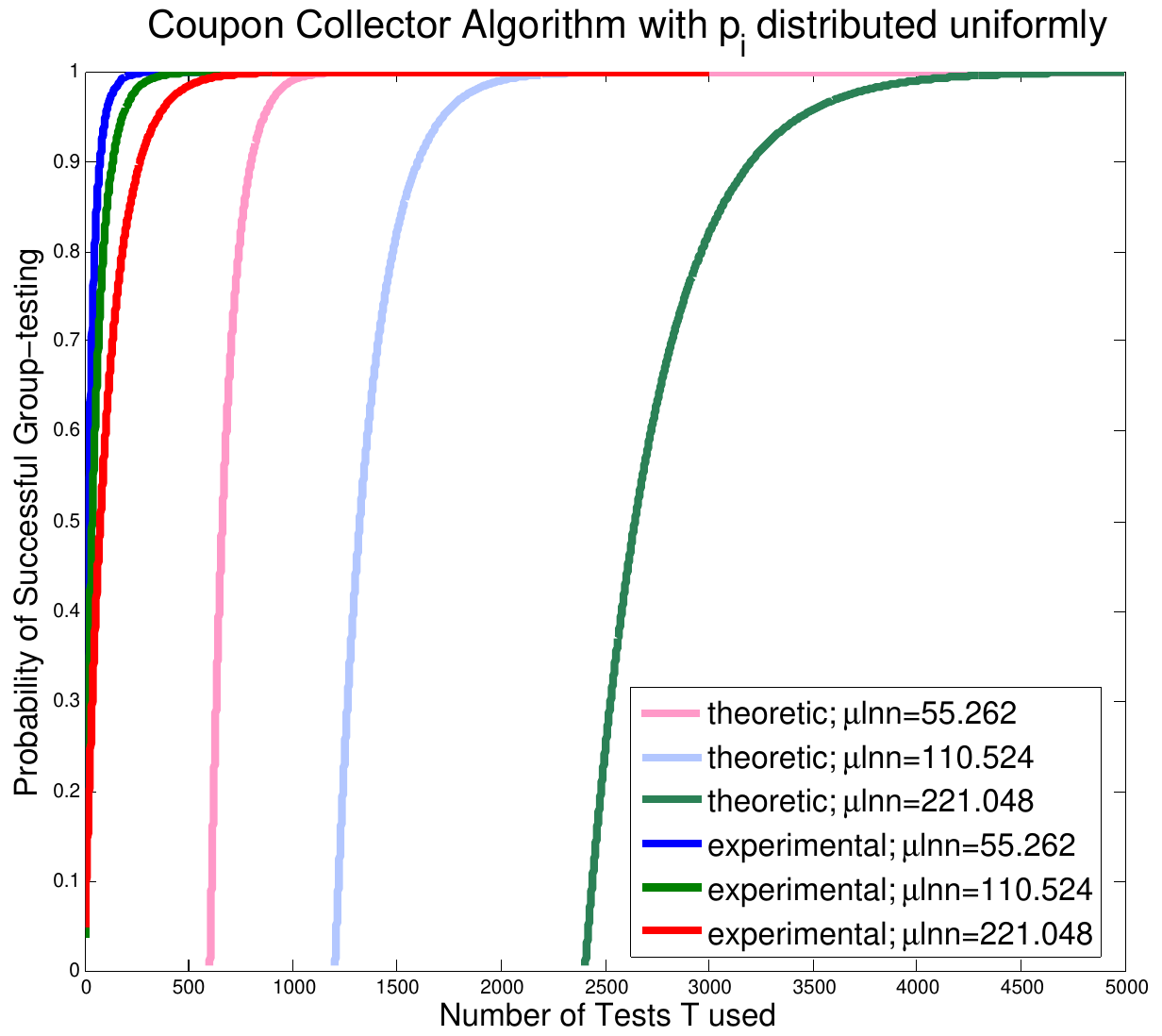}
      \label{fig:b-uni}}
  \caption{(a). The probability of successful group testing as a function of number of tests $T$ with exponentially distributed $\pvt$ and $n=1000$, $\mu=8,16,32$ sampled by independent $200$ trials. (b). The probability of successful group testing as a function of number of tests $T$ with linearly distributed $\pvt$ and $n=1000$, $\mu=8,16,32$ sampled by independent $200$ trials. (c). The probability of successful group testing as a function of number of tests $T$ with uniformly distributed $\pvt$ and $n=1000$, $\mu=8,16,32$ sampled by independent $200$ trials.}
\end{figure}

%


%
%
%
\clearpage
\addcontentsline{toc}{section}{Bibliography}
{\bibliographystyle{IEEEtran}
	\bibliography{ref}}

\begin{thebibliography}{10}
\providecommand{\url}[1]{#1}
\csname url@samestyle\endcsname
\providecommand{\newblock}{\relax}
\providecommand{\bibinfo}[2]{#2}
\providecommand{\BIBentrySTDinterwordspacing}{\spaceskip=0pt\relax}
\providecommand{\BIBentryALTinterwordstretchfactor}{4}
\providecommand{\BIBentryALTinterwordspacing}{\spaceskip=\fontdimen2\font plus
\BIBentryALTinterwordstretchfactor\fontdimen3\font minus
  \fontdimen4\font\relax}
\providecommand{\BIBforeignlanguage}[2]{{%
\expandafter\ifx\csname l@#1\endcsname\relax
\typeout{** WARNING: IEEEtran.bst: No hyphenation pattern has been}%
\typeout{** loaded for the language `#1'. Using the pattern for}%
\typeout{** the default language instead.}%
\else
\language=\csname l@#1\endcsname
\fi
#2}}
\providecommand{\BIBdecl}{\relax}
\BIBdecl

\bibitem{dorfman1943detection}
R.~Dorfman, ``The detection of defective members of large populations,''
  \emph{The Annals of Mathematical Statistics}, vol.~14, no.~4, pp. 436--440,
  1943.

\bibitem{du1993combinatorial}
D.~Z. Du and F.~Hwang, \emph{Combinatorial group testing and its
  applications}.\hskip 1em plus 0.5em minus 0.4em\relax World Scientific, 1993.

\bibitem{chan2011non}
C.~L. Chan, P.~H. Che, S.~Jaggi, and V.~Saligrama, ``Non-adaptive probabilistic
  group testing with noisy measurements: Near-optimal bounds with efficient
  algorithms,'' in \emph{Communication, Control, and Computing (Allerton), 2011
  49th Annual Allerton Conference on}.\hskip 1em plus 0.5em minus 0.4em\relax
  IEEE, 2011, pp. 1832--1839.

\bibitem{malioutov2012boolean}
D.~Malioutov and M.~Malyutov, ``Boolean compressed sensing: Lp relaxation for
  group testing,'' in \emph{Acoustics, Speech and Signal Processing (ICASSP),
  2012 IEEE International Conference on}.\hskip 1em plus 0.5em minus
  0.4em\relax IEEE, 2012, pp. 3305--3308.

\bibitem{du2006pooling}
D.~Du and F.~Hwang, \emph{Pooling designs and nonadaptive group testing:
  important tools for DNA sequencing}.\hskip 1em plus 0.5em minus 0.4em\relax
  World Scientific Pub Co Inc, 2006, vol.~18.

\bibitem{goodrich2008improved}
M.~T. Goodrich and D.~S. Hirschberg, ``Improved adaptive group testing
  algorithms with applications to multiple access channels and dead sensor
  diagnosis,'' \emph{Journal of Combinatorial Optimization}, vol.~15, no.~1,
  pp. 95--121, 2008.

\bibitem{wolf1985born}
J.~Wolf, ``Born again group testing: Multiaccess communications,''
  \emph{Information Theory, IEEE Transactions on}, vol.~31, no.~2, pp.
  185--191, 1985.

\bibitem{wadayama2013analysis}
T.~Wadayama, ``An analysis on non-adaptive group testing based on sparse
  pooling graphs,'' \emph{arXiv preprint arXiv:1301.7519}, 2013.

\bibitem{hagerup1990guided}
T.~Hagerup and C.~R{\"u}b, ``A guided tour of chernoff bounds,''
  \emph{Information processing letters}, vol.~33, no.~6, pp. 305--308, 1990.

\bibitem{chan2012non}
C.~L. Chan, S.~Jaggi, V.~Saligrama, and S.~Agnihotri, ``Non-adaptive group
  testing: Explicit bounds and novel algorithms,'' in \emph{Information Theory
  Proceedings (ISIT), 2012 IEEE International Symposium on}.\hskip 1em plus
  0.5em minus 0.4em\relax IEEE, 2012, pp. 1837--1841.

\bibitem{torney1998optimizing}
D.~Torney, F.~Sun, and W.~Bruno, ``Optimizing nonadaptive group tests for
  objects with heterogeneous priors,'' \emph{SIAM Journal on Applied
  Mathematics}, vol.~58, no.~4, pp. 1043--1059, 1998.

\bibitem{sobel1959group}
M.~Sobel and P.~A. Groll, ``Group testing to eliminate efficiently all
  defectives in a binomial sample,'' \emph{Bell System Technical Journal},
  vol.~38, no.~5, pp. 1179--1252, 1959.

\bibitem{Mal:78}
M.~B. Malyutov, ``Separating property of random matrices,'' \emph{Mat.
  Zametki}, vol.~23, pp. 155--167, 1978.

\bibitem{Dya:04}
A.~D'yachkov, ``Lectures on designing screening experiments,'' \emph{Lecture
  Note Series~10, Combinatorial and Computational Mathematics Center, Pohang
  University of Science and Technology (POSTECH)}, p. 112, Feb 2004.

\bibitem{shannon2001mathematical}
C.~E. Shannon, ``A mathematical theory of communication,'' \emph{ACM SIGMOBILE
  Mobile Computing and Communications Review}, vol.~5, no.~1, pp. 3--55, 2001.

\bibitem{huffman1952method}
D.~A. Huffman, ``A method for the construction of minimum-redundancy codes,''
  \emph{Proceedings of the IRE}, vol.~40, no.~9, pp. 1098--1101, 1952.

\bibitem{cover2012elements}
T.~M. Cover and J.~A. Thomas, \emph{Elements of information theory}.\hskip 1em
  plus 0.5em minus 0.4em\relax John Wiley \& Sons, 2012.

\bibitem{hoeffding1963probability}
W.~Hoeffding, ``Probability inequalities for sums of bounded random
  variables,'' \emph{Journal of the American statistical association}, vol.~58,
  no. 301, pp. 13--30, 1963.

\bibitem{boneh1989coupon}
A.~Boneh and M.~Hofri, ``The coupon-collector problem revisited,'' 1989.

\end{thebibliography}

\section{Appendix}

\subsection{Proof of Theorem~\ref{thm:low_bnd}}
\label{pf1}
\begin{proof}
The input vector $\ivt$, noiseless result vector $\rvt$ and estimated input vector $\ovt$ form a Markov chain 
$\ivt\to\rvt\to\ovt$. 
Moreover,
\begin{align}
\ent\left(\ivt\right) =&  \ent\left(\ivt|\ovt\right)+\mif\left(\ivt;\ovt\right).
\end{align}

Define an error random variable $E$ such that
\begin{align*}
E=&
\begin{cases}
1, & \text{if } \ovt\not=\ivt\\
0, & \text{if } \ovt=\ivt
\end{cases}.
\end{align*}

By Fano's inequality, we can bound the conditional entropy as
\begin{align*}
\ent\left(\ivt|\ovt\right) =& \ent\left(E,\ivt|\ovt\right)\\
=&\ent\left(E|\ovt\right)+\Pr\left[E=0\right]\ent\left(\ivt|\ovt,E=0\right)+\nonumber\Pr\left[E=1\right]\ent\left(\ivt|\ovt,E=1\right)\\
\leq&\ent\left(\et\right)+\et\ent\left(\ivt\right).
\end{align*}

Also we have $\mif\left(\ivt;\rvt\right)\leq\ent\left(\rvt\right)\leq \log_2|\rvt|=\nt$ by the data-processing inequality.
Hence we obtain that \[\nt \geq \left(1-\et\right)\ent\left(\ivt\right).\]
\end{proof}

\subsection{Proof of Theorem~\ref{th:LA}}
\label{pf2}
We prove the bound on $\expt\left[\nt\right]$ for both the maximum entropy-based construction (ME) and the Shannon-Fano/Huffman coding-based construction (S/H).

First, consider the  \textit{a priori} probabilities $\{p_i\}_{i=1}^{n}$ of items that are involved in a test at stage
$k$. For the ME, the group construction implies
\begin{align*}
p_i&\leq{\left(\frac{1}{2}\right)}^{k}.
\end{align*}
Therefore, the length of branch $\ell_i$ for each \textit{a priori} probability $p_i$ is bounded by
\begin{align}
\ell_i \leq \left\lceil\log_2{\frac{1}{p_i}}\right\rceil.\label{eq:adp0}
\end{align}

Note that Inequality (\ref{eq:adp0}) also holds for the Shannon-Fano coding~\cite{shannon2001mathematical}. To justify the Shannon-Fano coding is well-defined, we first introduce the following lemma:
\begin{lemma}
\label{thm:lemma1}
Let $k$ be a positive integer. If $0<p_i<1$ for $1\le i\le k$ and
$1/2\leq \prod_{i=1}^{k}{\left(1-p_{i}\right)}\leq 3/4$,
then $1/4\leq \sum_{i=1}^{k}p_i\leq1$.
\end{lemma}
\begin{proof}
Given $\prod_{i=1}^{k}{\left(1-p_{i}\right)}\geq 1/2$, or equivalently, 
\begin{equation}
\prod_{i=1}^{k} \frac{1}{1-p_{i}} \leq 2, \label{eq:adp1}
\end{equation}
it follows that the inequality (\ref{eq:adp1}) can then be expanded by its geometric sum as
\begin{align*}
2 & \leq \prod_{i=1}^{k}\left({\sum_{j=0}^{\infty}{{p_i}^j}}\right) \\
& \leq 1+\sum_{i=1}^{k}{p_i},
\end{align*}
which yields the desired inequality $\sum_{i=1}^{k}p_i\leq1$.

Furthermore, the Weierstrass product inequality implies that
\begin{align*}
\frac{3}{4}\geq \prod_{i=1}^{k}{\left(1-p_{i}\right)}\geq 1-\sum_{i=1}^{k}p_i
\end{align*}
yielding that
\begin{align}
\label{app:5}
\sum_{i=1}^{k}p_i\geq \frac{1}{4}.
\end{align}

\end{proof}

It remains to note that that under the partition for each subset $\ts_{1,1}^{r}$  such that
\begin{align}
 \prod_{X_{i}\in \ts_{1,1}^{r}}{\left(1-p_{i}\right)}\geq\frac{1}{2},\label{eq:ad2}
\end{align}
then the S/H-based algorithm is well-defined. Since (\ref{eq:ad2}) is the construction requirement in the first stage, by Lemma~\ref{thm:lemma1}, we have the Shannon-Fano/Huffman coding procedure is well-defined since the summation of \textit{a priori} probabilities within each subset is smaller or equal to $1$. Recall that $\overline{r}$ denotes the number of subsets $\ts_{1,1}^{r}$ that are to be tested in the first stage. It follows that 
\begin{align*}
\min_{r=1}^{\overline{r}}\sum_{i:X_i\in\ts_{1,1}^{r}}^{k}\overline{r} p_i \leq\sum_{i=1}^{n}p_i  = \mu.
\end{align*}

Using (\ref{app:5}), we get
\begin{align*}
\overline{r}\leq 4\mu.
\end{align*}

For each branch of length $\ell_i$, the number of tests required is at most $2\ell_i$ with probability $p_i$ when the corresponding item is defective. Therefore, we can bound the expected number of tests $\expt[\nt]$ as
\begin{align}
\label{eq:ad3}
\expt[\nt]&\leq \sum_{i=1}^n 2p_i\ell_i + \overline{r}\\
\nonumber
&\leq \sum_{i=1}^n 2p_i\left(\log_2\frac{1}{p_i}+1\right) + \overline{r}\\
\nonumber
&\leq 2\ent\left(\ivt\right)+6\mu
\end{align}
where $\mu$ is the summation of all \textit{a priori} probabilities and (\ref{eq:ad3}) comes from our testing procedure such that
a positive testing outcome, implies two more tests for both its children.

\subsection{Proof of Corollary~\ref{th:LA2}}
\label{pf2.5}
For the second part of the result of the adaptive algorithms, we show that $\nt$ concentrates ``properly''. The proof is based on partitioning the universal set $\tps$ into $K$ subsets $\is_1,\ldots,\is_{K}$. For each subset $\is_1,\ldots,\is_{K}$, denote by $T_1,\ldots,T_{K}$ the corresponding number of tests required and let $n_s$ be the number of items in the subset $\is_s$. From Theorem~\ref{th:LA}, we know that
\begin{align*}
1\leq &T_s\leq n_s\\
\expt[&T_s]\leq 2\ent\left(\ivt_s\right)+6\mu_s, \quad \forall \ s=1,\ldots,K
\end{align*}

Moreover, the random variables $T_1,\ldots,T_K$ are independent, since the items $X_1,\ldots,X_n$ are all independent. Also, $T=\sum_{s=1}^{K}T_s$. We have
\begin{align*}
\expt\left[T\right] \leq 2\sum_{s=1}^{K}\ent\left(\ivt_s\right)+6\sum_{s=1}^{K}\mu_s\leq 2\ent\left(\ivt\right)+6\mu.
\end{align*}
Applying  Hoeffding's inequality~\cite{hoeffding1963probability},
\begin{align}
\Pr\left(T \geq 2\left(1+\delta\right) \left(\ent\left(\ivt\right)+3\mu\right) \right)\leq \exp\left(-\frac{2K^2\delta^2}{\sum_{s=1}^{K}\left(n_s-1\right)^2}\right).
\end{align}

Setting $K=n^{3/4}$ and $n_s=n^{1/4}$ for all $s=1,\ldots,K$, we get
\begin{align*}
\Pr\left(T \geq 2\left(1+\delta\right) \left(\ent\left(\ivt\right)+3\mu\right) \right)\leq \exp\left(-2\delta^2n^{1/4}\right).
\end{align*}

\subsection{Proof of Theorem~\ref{th:ModCoco}}
\label{pf3}
\begin{proof}
The proof is modified from the proof of Thoerem 3 in~\cite{chan2011non}. The goal is to efficiently identify all non-defective items in the universal set $\tps$. As \cite{chan2011non} pointed out, it is possible to map the problem to the Coupon Collector's Problem. Non-defective items stand for the coupons. The set of negative tests which directly reveals non-defective items can be viewed as a chain of coupon collection.  

Then for each row, we assume a fixed \textit{group testing sampling parameter} $g>0$ and without of generality we assume that $g$ is an integer. We draw the coupons $g$ times (with replacement) according to a particular sampling distribution $\hpvt$ (specified in~\ref{eq:dis}). Hence the probability of obtaining an outcome $0$ for each test is $\left(\sum_{i=1}^{n}\widehat{p_{i}}\left(1-p_{i}\right)\right)^{g}$ and in total we draw the coupons, \textit{i.e.,} the non-defective items from the universal set $\tps$ $Tg$ times. Thus, we can regard a test as a length-$g$ sequence of selection and when a collector obtains a full set of coupons, the number of coupons collected should be at least the stopping time $\mathcal{T}$. In expectation, we can summarize the following equation:
\begin{align}
\nt g\left(\sum_{i=1}^{n}\widehat{p}_{i}\left(1-p_{i}\right)\right)^{g} &\geq \expt{[\mathcal{T}]}. \label{eq:et1}
\end{align}
For items being drawn with non-uniform distribution $\hpvt$,~\cite{boneh1989coupon} suggests that the expected stopping time $\expt{[\mathcal{T}]}$ is given by
\begin{equation}
\expt{[\mathcal{T}]} = \sum_{r=1}^{n}\left(-1\right)^{r+1}\sum_{1\leq i_{1}<\dotsm<i_{r}\leq
n}\frac{1}{\widehat{p}_{i_{1}}+\widehat{p}_{i_{2}}+\dotsm+\widehat{p}_{i_{r}}}. \label{eq:cc0}
\end{equation}

\begin{lemma}
\label{thm:bino1}
Let $n \in \mathbb{Z^{+}}$, we have
\begin{align*}
\sum_{r=1}^{n}\left(-1\right)^{r-1}{{n}\choose{r}}\frac{1}{r} &= \sum_{r=1}^{n}{\frac{1}{r}}.
\end{align*}
\end{lemma}
\begin{proof}
\begin{align}
\sum_{r=1}^{n}{\frac{1}{r}} & = \int_0^1 \frac{1-s^{n}}{1-s} ds \label{eq:bino0} \\
& = \int_0^1 \frac{1-\left(1-t\right)^{n}}{t} dt \label{eq:bino1} \\
& = \int_0^1 \left[\sum_{r=1}^{n}(-1)^{r-1}{{n}\choose{r}}t^{r-1}\right] dt \label{eq:bino2} \\
& = \sum_{r=1}^{n}\left(-1\right)^{r-1}{{n}\choose{r}}\frac{1}{r} \nonumber
\end{align}
where (\ref{eq:bino0}) follows from the expansion of geometric sum; (\ref{eq:bino1}) follows from substituting $s=1-t$, and (\ref{eq:bino2}) follows from the binomial theorem.
\end{proof}
Recall that $\mu:=\sum_{i=1}^{n}{p_{i}}$. Inequality (\ref{eq:et1}) can be further computed as
\begin{align}
\expt{[\mathcal{T}]} & = (n-\mu)\sum_{r=1}^{n}\left(-1\right)^{r+1}  \sum_{1\leq i_{1}<...<i_{r}\leq n}\frac{1}{r}\left( 1-\frac{p_{i_{1}}+p_{i_{2}}+...+p_{i_{r}}}{r} \right)^{-1} \label{eq:cc1} \\
& = (n-\mu)\sum_{r=1}^{n}\left(-1\right)^{r+1}  \sum_{1\leq i_{1}<...<i_{r}\leq n}\frac{1}{r}\left(\sum_{j=0}^{\infty}\left(\frac{p_{i_{1}}+p_{i_{2}}+...+p_{i_{r}}}{r}\right)^{j}\right)\label{eq:cc2} \\
&\leq (n-\mu)\sum_{r=1}^{n}\left(-1\right)^{r+1}{{n}\choose{r}}\frac{1}{r}\left ( 1 + \frac{\mu}{n}\sum_{s=0}^{\infty} 2^{-s} \right ) \label{eq:cc3} \\
& = \left(n-\mu\right)\ln n \left (1+\frac{2 \mu}{n} \right ) \label{eq:cc4} \\
& < \left(n+\mu\right)\ln n \label{eq:cc5}
\end{align}
where (\ref{eq:cc1}) follows from substituting $\widehat{p_{i}}:=(1-p_{i})/(n-\mu)$; (\ref{eq:cc2}) follows from the expansion of geometric sum with the fact that every $p_{i}$ as well as the average $r^{-1}\sum_{j=1}^{r}{p_{i_{j}}}$ is between 0 and 1. Since we assume the universal set $\tps$ is bounded above by ${1}/{2}$, making use of  $p_i<{1}/{2}$ and expanding Eqn~(\ref{eq:cc2}) we obtain (\ref{eq:cc3}). Moreover, (\ref{eq:cc4}) follows from lemma~\ref{thm:bino1} and $\sum_{r=1}^{n} 1/r \leq \ln(n)$.

Substituting (\ref{eq:cc5}) into (\ref{eq:et1}) and optimizing for $g$, we obtain $g^{*}=-1\big/\ln\left(\sum_{i=1}^{n}\widehat{p}_{i}\left(1-p_{i}\right)\right)$. 
Such choice of $g^{*}$ and the assumption $\mu\ll n$ allow (\ref{eq:et1}) to be simplified as 
\begin{equation}
\nt \geq e\mu\ln n
\label{eq:et2}
\end{equation}
since the ratio between the expected stopping time and the expected non-defective items in a single negative test can be computed as
\begin{align*}
\frac{\left(n+\mu\right)\ln n}{g^{*}\left(\sum_{i=1}^{n}\widehat{p_{i}}(1-p_{i})\right)^{g^{*}}}
&=\left(n+\mu\right)\ln n\left(-\ln\left(\sum_{i=1}^{n}\widehat{p_{i}}(1-p_{i})\right)\right)\\
&=e\left(n+\mu\right)\ln n\left(-\ln\left(\frac{\sum_{i=1}^{n}(1-p_{i})^{2}}{n-\mu}\right)\right)\\
& =e\left(n+\mu\right)\ln n\left(-\ln\left(\frac{\sum_{i=1}^{n}(1-2p_{i}+{p_{i}}^{2})}{n-\mu}\right)\right)\\
& <e\left(n+\mu\right)\ln n\left(\ln\left(\frac{n-\mu}{n-2\mu}\right)\right)\\
&  \doteq e\mu \ln n.
\end{align*}
Note that (\ref{eq:et1}) only accounts for the expectation. Now we take variance in consideration. By Chernoff bound, the actual number of items in the negative tests can be smaller than $1-\alpha$ times the expected number with probability at most $\exp\left(-\alpha^{2}T\right)$. In tail estimate of the coupon collector problem, with probability $n^{{-\beta}/{2}+1}$, a collector requires more than $\beta\expt{[\mathcal{T}]}$ coupons before he is able to collect a full set. Thus, applying the union bound over two error events, the Inequalities~(\ref{eq:et1}) and~(\ref{eq:et2}) are generalized as the following statement:
\begin{align}
\left(1-\alpha\right)\nt \geq& {e\beta}\mu \ln n \label{eq:cc6},
\end{align}
which does not hold with probability $\et$ at most $\exp\left(-\alpha^{2}T\right)+n^{\frac{-\beta}{2}+1}$. Taking $\alpha=\frac{1}{2}$ in (\ref{eq:cc6}), we can bound the probability of error as
\begin{align*}
\et \leq& \exp\left(-\alpha^{2}T\right)+n^{\frac{-\beta}{2}+1}\\
\leq& \exp\left(-\frac{\nt}{4}\right)+n^{\frac{-\beta}{2}+1}\\
\leq& n^{-\frac{e\mu\beta}{2}}+n^{\frac{-\beta}{2}+1}\\
\leq& 2n^{\frac{-\beta}{2}+1}.
\end{align*}
If we reparameterize  $\beta$ by $2(\delta+1)$, we get $2n^{\frac{-\beta}{2}+1}= 2n^{-\delta}$. Hence, Theorem~\ref{th:ModCoco} holds.
\end{proof}

\subsection{Proof of Theorem~\ref{th:BlockCoco}}
\label{pf4}

\begin{proof}

Before proceeding, we need some additional notations based on Definition~\ref{def:2} and~\ref{def:3}. Let $A\leq L$ be the total number of ample subsets and denote by $\is_{s}$ an ample subset indexed by $s \in \{1, 2,\ldots, A \}$. Moreover, let  $n_{s}$ be the number of items in
$\is_s$, $\mu_{s}:=\sum_{i\in\is_{s}}p_i$ be the sum of \textit{a priori} probabilities of items in $\is_{s}$ and $\ivt_s$ be the population vector for each $\is_s$.

The total number of tests $\nt$ is the sum of the total number of tests for ample subsets, denoted by $\nt_{\mathrm{ample}}$, and the total number of tests for unbounded or non-ample subsets, denoted by $\nt_{\mathrm{non}}$. 

First, for the non-ample subsets, the number of tests required is at most $\left(L-2\right)\left(\Gamma_{\gamma}-1\right)+2\mu$ (at most $\Gamma_{\gamma}-1$ tests for less than or equal to $L-2$ subsets, together with the items in $\is_L$). According to the pre-partition model assumption, the number of subsets $L$ satisfies
\begin{align*}
\left(\frac{\uet}{2\np}\right)^{\left(\frac{1}{2}\right)^{L-1}}>\gamma
\end{align*}
implying that
\begin{align*}
L-2\leq \log_{2}\left(\log_{1/\gamma}\left(\frac{2\np}{\uet}\right)\right)=: \Gamma_{\gamma}
\end{align*}

Thus, we can bound $\left(L-2\right)\left(\Gamma_{\gamma}-1\right)$ by
\begin{align}
\label{app:8}
\left(L-2\right)\left(\Gamma_{\gamma}-1\right) & \leq \Gamma_{\gamma}^2
\end{align}
yielding that
\begin{equation}
\nonumber
\nt_{\mathrm{non}} \leq\Gamma_{\gamma}^2+2\mu.
\end{equation}

Second, for the ample subsets, by Theorem~\ref{th:ModCoco}, with probability of error at most $2{n_{s}^{-\delta}}$ for each ample subset $\is_{s}$, $\nt_1$ can be bounded as
\begin{align}
\label{app:6}
\nt_{\mathrm{ample}} & \leq 4e\left(1+\delta\right)\sum_{s=1}^{A}{\mu_{s} \ln n_{s}}.
\end{align}


Denote by $p^{\max}_s:=\max_{i:X_i\in\is_s}p_i$ the maximal probability in $\{p_i\}_{i:X_i\in\is_s}$. Furthermore, according to pre-partition model (see Figure~\ref{fig:partition}), for each subset $\is_s$ ($s=2,\ldots,L$),
\begin{align}
\label{app:fux}
p^{\max}_s\leq \left(\frac{\uet}{2n}\right)^{\left(1/2\right)^{(s-1)}}.
\end{align}

Thus, the entropy $\ent\left(\ivt_s\right)$ can be bounded by
\begin{align}
\label{app:1}
\ent\left(\ivt_s\right)&=\sum_{i=1}^{n}p_i\log_2\frac{1}{p_i}
\geq \sum_{i:X_i\in\is_s}p_i\log_2\frac{1}{p_i}
\geq \sum_{i:X_i\in\is_s}p_i\log_2\frac{1}{p^{\max}_s}
=\mu_s\log_2\frac{1}{p^{\max}_s}.
\end{align}
where (\ref{app:1}) follows since $X_1,\ldots,X_n$ are independent and we define $\mu_s:=\sum_{i:X_i\in\is_s}p_i$. Putting (\ref{app:fux}) into (\ref{app:1}), we obtain
\begin{align}
\label{app:3}
\ent\left(\ivt_s\right)\geq \left(\frac{1}{2}\right)^{s-1}\mu_s\log_2\left(\frac{2n}{\uet}\right)
\end{align}

Moreover, since $L\leq\Gamma_{\gamma}$, (\ref{app:3}) implies that
\begin{align}
\nonumber
\ent\left(\ivt_s\right)\geq&\left(\frac{1}{2}\right)^{L-3}\mu_s\log_2\left(\frac{2n}{\uet}\right)\\
\nonumber
\geq& \left(\frac{1}{2}\right)^{\Gamma_{\gamma}-3}\mu_s\log_2\left(\frac{2n}{\uet}\right)\\
\nonumber
\geq& 4 \left(\log_{1/\gamma}\frac{2n}{\uet}\right)^{-1}\mu_s\log_2\left(\frac{2n}{\uet}\right)\\
\label{app:4}
=& 4\log_2\left(\frac{1}{\gamma}\right)\mu_s.
\end{align}

Combining (\ref{app:6}) and (\ref{app:4}),
\begin{align}
\label{app:7}
\nt_{\mathrm{ample}} \leq 4e\left(1+\delta\right)\ln n\sum_{s=2}^{L-2}{\mu_{s} }\leq \frac{e\ln n}{\log_2(1/\gamma)}\left(1+\delta\right)\sum_{s=2}^{L-2}\ent\left(\ivt_s\right)\leq\frac{e\ln n}{\log_2(1/\gamma)}\left(1+\delta\right)\ent\left(\ivt\right).
\end{align}

The probability that there exists a misclassification in the first subset $\is_1$ is bounded from above by $\sum_{i:X_i\in\is_1}p_i\leq {\uet}/{2}$. Applying the union bound and putting (\ref{app:6}) and (\ref{app:7}) together, we conclude that  the total number of tests is bounded by
\begin{align*}
\nt & = \nt_{\mathrm{ample}} + \nt_{\mathrm{non}}\leq \frac{e\ln n}{\log_2(1/\gamma)}\left(1+\delta\right)\ent\left(\ivt\right)+\Gamma_{\gamma}^2+ 2\mu
\end{align*}
with the probability of error $\et$ satisfying (including the stage when setting the items zero directly if the corresponding probabilities are small)
\begin{align*}
\et \leq&  \sum_{s=1}^{A}{n_{s}^{-\delta}}+\frac{1}{2}\uet.
\end{align*}

Since the subsets $\is_s$ are ample, \textit{i.e.,} $n_s=\left|\is_s\right|>\Gamma_\gamma$,
\begin{align*}
\et \leq& A{\Gamma_\gamma^{-\delta}}+\frac{1}{2}\uet.
\end{align*}

Moreover, since $A\leq L<\Gamma_\gamma$ in agreement with (\ref{app:8}),
\begin{align*}
\et \leq& {\Gamma_\gamma^{-\delta+1}}+\frac{1}{2}\uet.
\end{align*}
\end{proof}


\end{document}